\ifdefined\XeTeXversion\else\pdfoutput=1\fi
\documentclass[11pt]{article}

\usepackage[margin=1in]{geometry}
\usepackage{iftex}
\ifPDFTeX
  \usepackage[T1]{fontenc}
  \usepackage{lmodern}
\fi
\usepackage{amsmath,amssymb,amsthm}
\usepackage{booktabs}
\usepackage{array}
\usepackage{tikz}
\usetikzlibrary{arrows.meta,positioning,plotmarks}
\usepackage{xcolor}
\usepackage[colorlinks=true,linkcolor=blue!55!black,citecolor=blue!55!black,urlcolor=blue!55!black]{hyperref}
\usepackage{xurl}
\hypersetup{pdftitle={Improving the lower bound on the Laughlin gap: An SDP search},pdfauthor={Qiang Gao}}

\newtheorem{theorem}{Theorem}[section]
\newtheorem{corollary}[theorem]{Corollary}
\newtheorem{lemma}[theorem]{Lemma}
\newtheorem{proposition}[theorem]{Proposition}
\theoremstyle{definition}
\newtheorem{definition}[theorem]{Definition}
\theoremstyle{remark}
\newtheorem{remark}[theorem]{Remark}

\numberwithin{equation}{section}

\newcommand{\Fock}{\mathcal{F}}
\newcommand{\Lift}{\mathcal{L}}
\newcommand{\Wh}{\widehat{W}}
\newcommand{\R}{\mathbb{R}}
\newcommand{\Z}{\mathbb{Z}}
\newcommand{\tr}{\operatorname{tr}}
\newcommand{\Bin}{\operatorname{Bin}}
\newcommand{\pbar}{\bar p}
\newcommand{\jbar}{\bar\jmath}
\newcommand{\ibar}{\bar\imath}
\newcommand{\tbar}{\bar t}
\newcommand{\GammaNum}{1401624220840445441298243497}
\newcommand{\GammaDen}{64\cdot 10^{26}}
\newcommand{\GammaDec}{0.2190037845063196}
\newcommand{\GammaOld}{0.0461673}

\title{Improving the lower bound on the Laughlin gap:\\ An SDP search}
\author{Qiang Gao}
\date{\today}

\begin{document}
\maketitle

\begin{abstract}
A recent preprint by OpenAI resolved the spherical fermionic Laughlin spectral-gap conjecture for the full $V_1$ interaction by a Fock-space inequality: $H_Q^2\ge \gamma H_Q$ on the lowest-Landau-level Fock space for all large flux $Q$ and every $\gamma$ below $0.0461673\ldots$, so that the gap above the Laughlin state at filling $1/3$ is at least $1/25$. That proof, with its original constant, has also been formalized in Lean. We show that the same argument supports a much larger constant. Two sharper estimates, a product-twirl bound for the three-body tail and a staircase price for the four-body allowance, raise the constant above $0.0601$ without changing the original certificate. Treating the auxiliary operators as the variables of a semidefinite program (SDP) and enlarging the finite set of orbital modes on which they act then gives $H_Q^2\ge\gamma H_Q$ for every $\gamma<\gamma_*=0.2190037845063196\ldots$, certified by exact rational arithmetic. Consequently the gap above the Laughlin state is at least $1/5$ for all sufficiently large spheres, and the untruncated planar $V_1$ Hamiltonian has positive spectrum bounded below by $\gamma_*$, uniformly in particle number and angular momentum. A heuristic extrapolation of the SDP optima suggests that auxiliary operators of the present form can give at most about $0.222$, while exact diagonalization of finite systems places the true gap near $0.4$. A larger family of auxiliary operators could push the bound further towards the true gap, and the SDP formulation makes such families easy to test. The proof of the improvement and the SDP search were completed by the AI model Claude Opus 5.5 (Anthropic) under the author's direction.
\end{abstract}

\section{The bound and its improvement}

The $V_1$ pseudopotential Hamiltonian is the simplest parent Hamiltonian of Laughlin's state at filling $1/3$ \cite{Laughlin1983,Haldane1983}: every pair of fermions pays energy one in relative angular momentum one and nothing otherwise. That its positive spectrum stays uniformly away from zero is a long-standing conjecture, supported by finite-size studies since the work of Haldane and Rezayi \cite{HaldaneRezayi1985} and recorded in Rougerie's survey \cite[Appendix~A]{Rougerie2019}. On the round sphere it has now been proved \cite{FockLaughlin}, and that proof, with its original constant, has been formalized in Lean \cite{FockLaughlinLean}. The proof establishes the operator inequality $H_Q^2\ge\gamma H_Q$ on the entire lowest-Landau-level Fock space, by comparing the normal-ordered square of $H_Q$ with a rotation-averaged sum of squares of auxiliary operators, and it obtains every $\gamma$ below
\[
  \gamma_*^{\mathrm{old}}=\frac{4616733319001}{10^{14}}=\GammaOld\ldots,
\]
whence a gap of at least $1/25$ above the Laughlin state.

The purpose of this note is quantitative. We keep the architecture of \cite{FockLaughlin}: the normal-ordered square, the spin-block comparisons, the exact diagonal change of variables that transfers four-body positivity to finite spheres, and the summable three-body tail. We improve each of the constants that enter the final bound, and we enlarge the finite set of orbital modes on which the auxiliary operators act. The resulting constant is about $4.7$ times larger.

\subsection{Normalization and main results}

We use the conventions of \cite{FockLaughlin} throughout. For an integer $Q\ge1$, $U_Q$ is the spin-$Q/2$ representation of $SU(2)$, the one-particle lowest Landau level on the sphere with $Q$ flux quanta. For $2\le N\le Q+1$,
\begin{equation}\label{eq:H}
  H_{N,Q}=\sum_{1\le i<j\le N}P^{(1)}_{ij}\Big|_{\bigwedge^N U_Q},
\end{equation}
where $P^{(1)}_{ij}$ projects onto pair spin $Q-1$ in tensor positions $i,j$, with coefficient one for each unordered pair. With $H_{0,Q}=H_{1,Q}=0$ put
\[
  \Fock_Q=\bigoplus_{N=0}^{Q+1}{\textstyle\bigwedge^N}U_Q,\qquad H_Q=\bigoplus_{N=0}^{Q+1}H_{N,Q},
\]
and
\[
  \gamma_*=\frac{\GammaNum}{\GammaDen}=\GammaDec\ldots>\frac15 .
\]

\begin{theorem}\label{thm:main}
For every $0<\gamma<\gamma_*$ there is an integer $Q_\gamma$ such that for every integer $Q\ge Q_\gamma$,
\begin{equation}\label{eq:main}
  H_Q^2\ \ge\ \gamma H_Q\qquad\text{on }\Fock_Q .
\end{equation}
The threshold $Q_\gamma$ is independent of particle number.
\end{theorem}

As in \cite{FockLaughlin}, the spectrum of $H_Q$ is then contained in $\{0\}\cup[\gamma,\infty)$. At the Laughlin flux the zero mode is unique \cite[Lemma~6.2]{FockLaughlin}, which gives the usual ground-state gap.

\begin{corollary}\label{cor:laughlin}
There is an integer $N_0\ge2$ such that, for every $N\ge N_0$ and $Q=3(N-1)$,
\begin{equation}\label{eq:gap}
  H_{N,Q}\ \ge\ \frac15\,\bigl(I-P_{\mathrm L,N}\bigr)\qquad\text{on }{\textstyle\bigwedge^N}U_Q,
\end{equation}
where $P_{\mathrm L,N}$ projects onto the line spanned by $\Psi_{\mathrm L,N}=\prod_{i<j}(u_iv_j-u_jv_i)^3$.
\end{corollary}

\begin{corollary}\label{cor:planar}
For every $N\ge2$ the untruncated planar $V_1$ Hamiltonian $H_N^{\mathrm{pl}}$ of \cite[Section~7]{FockLaughlin} satisfies $(H_N^{\mathrm{pl}})^2\ge\gamma_*H_N^{\mathrm{pl}}$. Equivalently $H_N^{\mathrm{pl}}\ge\gamma_*(I-P^{\mathrm{pl}}_{0,N})$, where $P^{\mathrm{pl}}_{0,N}$ projects onto $\ker H_N^{\mathrm{pl}}$, uniformly in $N$ and in the total angular momentum. In particular the fermionic case $\ell=3$ (filling $1/3$) of the spectral-gap conjecture \cite[Conjecture~A.1]{Rougerie2019} holds, in its pair-projector normalization, with constant $\gamma_*>1/5$.
\end{corollary}

\begin{corollary}\label{cor:charge}
With $N_0$ as in Corollary~\ref{cor:laughlin}, $Q_N=3(N-1)$ and $E_0(n;Q_N)$ the lowest eigenvalue of $H_{n,Q_N}$, the charge gap $\Delta_N=E_0(N+1;Q_N)+E_0(N-1;Q_N)-2E_0(N;Q_N)$ satisfies $\Delta_N\ge N/[5(N-1)]$ for $N\ge N_0$.
\end{corollary}

The thresholds are existential, exactly as in \cite{FockLaughlin}. Our margins are small, so the flux at which the inequality first holds is presumably very large; the statements concern the thermodynamic limit.

\subsection{Where the original constant was lost}

The original constant is the balance
\begin{equation}\label{eq:ledger}
  \gamma_*^{\mathrm{old}}=1-\eta-\delta-1222\,\varepsilon,\qquad
  \eta=0.93527\ldots,\quad \delta=\tfrac{61}{4096}=0.01489\ldots,\quad 1222\,\varepsilon=0.003666 .
\end{equation}
Here $\eta=\sum_\rho\lambda_\rho^2$ is the two-body cost of the pair-annihilator parts of the seven auxiliary operators, $\delta$ bounds the three-body spin blocks not reached by those operators (a Cauchy--Schwarz estimate with a factor $z+1$ in block $z$), and $1222\,\varepsilon$ pays for an allowance $\varepsilon I$ in the four-body certificate (a Bessel estimate repeated in each of the $23$ retained spin blocks). Each term can be reduced.

\begin{itemize}
\item \emph{Three-body tail} (Section~\ref{sec:tail}). Comparing a spin block with the rotation average of a single product vector replaces the factor $z+1$ by $1/\max_p\Bin(z,\tfrac13)(p)\sim\tfrac23\sqrt{\pi z}$. The tail $\delta$ drops from $0.0149$ to $0.0041$.
\item \emph{Four-body allowance} (Section~\ref{sec:allowance}). The rotation average of one vector at total deficit $T$ pays the allowance of \emph{every} spin block $D\le T$ at once. The identity allowance of \cite{FockLaughlin} therefore costs $144\,\varepsilon$ rather than $1222\,\varepsilon$, and block-dependent allowances cost $\sum_D\lceil D/2\rceil\max_{D'\ge D}\varepsilon_{D'}$, which we call the staircase price.
\item \emph{The auxiliary operators} (Section~\ref{sec:sdp}). Every quantity in the certificate depends on the auxiliary operators only through their Gram matrices, linearly. Optimizing over them is a semidefinite program. Within the original index sets this already lowers $\eta$ to $0.92144$ with no allowance at all (Appendix~\ref{app:ten}). Enlarging the index sets is the main lever: it lowers $\eta$ to $0.78097$ and simultaneously pushes the tail beyond block $35$, where it is negligible.
\end{itemize}

Table~\ref{tab:ledger} records the effect of each step. Every entry is a rigorous constant: those that use new auxiliary operators are certified by exact rational arithmetic, with certificates printed in Appendix~\ref{app:cert} or distributed as ancillary files, and the others use the original certificate unchanged.

\begin{table}[ht]
\centering
\renewcommand{\arraystretch}{1.15}
\begin{tabular}{@{}llr@{}}
\toprule
Auxiliary operators & Estimates and allowances & $\gamma$ \\
\midrule
original $7$ rows \cite{FockLaughlin} & as in \cite{FockLaughlin} & $0.0461673$\\
original $7$ rows & staircase price (Cor.~\ref{cor:144}) & $0.0494013$\\
original $7$ rows & product-twirl tail (Lemma~\ref{lem:tail}) & $0.0569644$\\
original $7$ rows & both & $0.0601984$\\
$10$ rows, original index sets & no allowance needed & $0.0636643$\\
$10$ rows, original index sets & product-twirl tail & $0.0744614$\\
$20$ rows, $(\pbar,\jbar,\ibar,\tbar)=(9,10,13,9)$ & no allowance needed & $0.1882970$\\
$37$ rows, $(\pbar,\jbar,\ibar,\tbar)=(14,15,15,14)$ & one block allowance & $0.2138574$\\
$47$ rows, $(\pbar,\jbar,\ibar,\tbar)=(16,17,17,16)$ & one block allowance & $0.2178926$\\
$64$ rows, $(\pbar,\jbar,\ibar,\tbar)=(17,18,18,17)$ & seven block allowances & $0.2190037$\\
\bottomrule
\end{tabular}
\caption{Rigorous constants $\gamma$ for which $H_Q^2\ge(\gamma-o(1))H_Q$ on $\Fock_Q$. The index bounds $(\pbar,\jbar,\ibar,\tbar)$ are defined in Section~\ref{sec:framework}; the original ones are $(7,8,8,7)$. The last line is Theorem~\ref{thm:main}; pricing its allowances with Proposition~\ref{prop:four}(b), still with $c_z=z+1$, raises it to $0.2190169$. Decimal values are truncated.}
\label{tab:ledger}
\end{table}

\subsection{Strategy and organization}

Section~\ref{sec:framework} restates the argument of \cite{FockLaughlin} for arbitrary finite index sets and arbitrary block-dependent allowances. Its main result, Theorem~\ref{thm:framework}, reduces the operator inequality~\eqref{eq:main} to finitely many limiting matrix inequalities, whatever the index sets; its proof follows \cite{FockLaughlin} line by line and we only record the points where index ranges enter. Sections~\ref{sec:tail} and~\ref{sec:allowance} prove the two new estimates. Section~\ref{sec:sdp} formulates the choice of auxiliary operators as a semidefinite program, describes how its numerical solutions are turned into exact certificates, and reports the results. Section~\ref{sec:proofs} proves the results stated above. Appendix~\ref{app:cert} specifies the certificates and their verification, Appendix~\ref{app:data} lists the ancillary data and code, and Appendix~\ref{app:human} describes the human input.

Figure~\ref{fig:flow} summarizes the logic. Only the first box involves numerics, and its output is checked exactly.

\begin{figure}[ht]
\centering
\begin{tikzpicture}[
  box/.style={draw=blue!55!black, rounded corners=1pt, align=center, inner sep=5pt, font=\small, text width=4.1cm},
  arr/.style={-{Latex[length=2mm]}, thick}]
\node[box] (a) {\textbf{Semidefinite program}\\ Gram matrices of the\\ auxiliary operators};
\node[box, right=0.6cm of a] (b) {\textbf{Exact certificate}\\ Rational rows; all\\ conditions checked exactly};
\node[box, right=0.6cm of b] (c) {\textbf{Theorem~\ref{thm:framework}}\\ Transfer to large spheres\\ and all particle numbers};
\draw[arr] (a) -- (b);
\draw[arr] (b) -- (c);
\end{tikzpicture}
\caption{The certificate is found numerically and verified in rational arithmetic. Theorem~\ref{thm:framework}, which rests on the uniform transfer argument of \cite{FockLaughlin}, converts it into the Fock-space inequality.}
\label{fig:flow}
\end{figure}
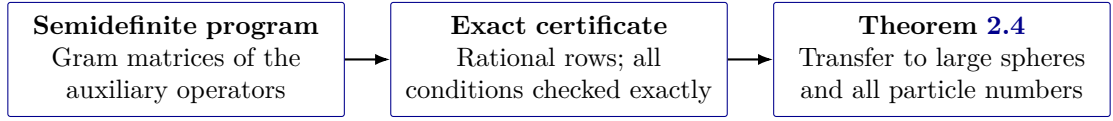

\section{The comparison argument for general index sets}\label{sec:framework}

\subsection{Notation}

We recall the objects of \cite[Sections~2--4]{FockLaughlin}. On $\Fock_Q$, $B(v)$ is the adjoint of left wedge multiplication by $v$, so that $B(v\wedge w)=B(w)B(v)$, and $c_j=B(e_j)$ for the normalized lowering basis $e_0,\dots,e_Q$ of $U_Q$. The lift of a matrix on $\bigwedge^kU_Q$ is $\Lift_k(|v\rangle\langle w|)=B(v)^\dagger B(w)$; it preserves order and intertwines rotations. With $a=Q-1$, $V_Q\subset\bigwedge^2U_Q$ is the spin-$a$ summand, $v_0,\dots,v_{2Q-2}$ its normalized lowering basis with $v_0=-e_0\wedge e_1$, $B_p=B(v_p)$ and
\[
  H_Q=\Lift_2(\Pi_{V_Q})=\sum_{p=0}^{2Q-2}B_p^\dagger B_p,\qquad d_a=2Q-1 .
\]
The wedge maps are $W_3:V_Q\otimes U_Q\to\bigwedge^3U_Q$, $W_4:V_Q\otimes V_Q\to\bigwedge^4U_Q$ and $\Wh:V_Q\otimes\bigwedge^2U_Q\to\bigwedge^4U_Q$. In $V_Q\otimes U_Q$, $R_{3,z}$ is the projection onto spin $3Q/2-1-z$, of dimension $d_{3,z}=3Q-1-2z$, and
\begin{equation}\label{eq:qz}
\begin{gathered}
  W_3^\dagger W_3-I=\sum_{z=0}^Q q_z(Q)R_{3,z},\\
  q_z(Q)=(-1)^z\frac{(Q)_z}{(2Q-2)_z}\Bigl(3z-1-\frac{z(z+1)}{Q}\Bigr)\ \longrightarrow\ q_z^\infty:=(3z-1)\bigl(-\tfrac12\bigr)^z,
\end{gathered}
\end{equation}
with $W_3R_{3,z}=0$ for $z\in\{0,1,3\}$ \cite[Lemma~3.2]{FockLaughlin}. The limiting coupling coefficients $s^{(u)}_{z,T,p}$ are the coefficients of $X^pY^{T-p}/\sqrt{p!(T-p)!}$ in $(uX-vY)^z(vX+uY)^{T-z}/\sqrt{z!(T-z)!}$, where $v=\sqrt{1-u^2}$ and $X$ refers to the first tensor factor \cite[Lemma~3.3]{FockLaughlin}. In $V_Q\otimes\bigwedge^2U_Q$, the copies of spin $2Q-1-D$ are labeled by odd $r\in\mathcal R_D=\{1,3,\dots\}\cap[1,D]$, with vectors $|D,T;r\rangle_Q$ at total deficit $T$, dimension $d_D=4Q-1-2D$, limiting coefficients $S_r(D,T;p,j,k)$ \cite[(23)]{FockLaughlin}, and $\chi_{rs}=\mathbf 1_{r=s=1}$. We write $w^Q_{Dr}=\Wh|D,D;r\rangle_Q$ for the wedged images of the highest vectors, $G^Q_D$ for their Gram matrix, and $E^Q_D$ for the finite-flux copy matrix of the four-body term \cite[(24)]{FockLaughlin}. Finally
\[
  n_D=\#\{(p,j,k):p+j+k=D,\ 0\le j<k\}=\Bigl\lfloor\frac{(D+1)^2}{4}\Bigr\rfloor,\qquad n_D-n_{D-1}=\Bigl\lceil\frac D2\Bigr\rceil .
\]

\subsection{Auxiliary operators and limiting data}

\begin{definition}\label{def:ansatz}
Fix integers $\pbar,\jbar,\ibar,\tbar\ge0$, and for $0\le t\le\tbar$ let
\[
  S_t=\{(p,j)\in\Z^2:\ 0\le p\le\pbar,\ 0\le j\le\jbar,\ 0\le p+j-t\le\ibar\}.
\]
A \emph{row} $\rho$ consists of a \emph{shift} $t_\rho\in\{0,\dots,\tbar\}$, a real $\lambda_\rho$ and reals $\alpha^\rho_{pj}$, $(p,j)\in S_{t_\rho}$. Its auxiliary operator is
\begin{equation}\label{eq:F}
  F_\rho=\lambda_\rho B_{t_\rho}+\sum_{(p,j)\in S_{t_\rho}}\alpha^\rho_{pj}\,c^\dagger_{p+j-t_\rho}c_jB_p .
\end{equation}
Put $T_3=\max\{p+j:(p,j)\in S_t,\ t\le\tbar\}$ and $T_4=\max\{p+j+q+l-t:(p,j),(q,l)\in S_t,\ t\le\tbar\}$.
\end{definition}

The original choice is $(\pbar,\jbar,\ibar,\tbar)=(7,8,8,7)$, for which $T_3=15$ and $T_4=23$. For the index bounds used below, $\ibar\le\pbar+\jbar\le\tbar+\ibar$, and then $T_3=\pbar+\jbar$ and $T_4=\pbar+\jbar+\ibar$.

\begin{definition}\label{def:data}
For a finite family of rows, define the limiting data
\begin{align}
  \eta&=\sum_\rho\lambda_\rho^2, \label{eq:eta}\\
  e_z&=\sum_\rho\sum_{T=\max(z,t_\rho)}^{T_3}\Bigl[2\lambda_\rho s_{z,T,t_\rho}\sum_{p+j=T}\alpha^\rho_{pj}s_{z,T,p}+\Bigl(\sum_{p+j=T}\alpha^\rho_{pj}s_{z,T,p}\Bigr)^2\Bigr],\qquad s=s^{(1/\sqrt3)}, \label{eq:ez}\\
  (E_D)_{rs}&=-\sum_\rho\sum_{(p,j),(q,l)\in S_{t_\rho},\ T\ge D}\alpha^\rho_{pj}\alpha^\rho_{ql}\,S_r(D,T;p,j,k)\,S_s(D,T;q,l,i), \label{eq:ED}
\end{align}
where in~\eqref{eq:ED} $i=p+j-t_\rho$, $k=q+l-t_\rho$ and $T=p+j+k$. Let $G_D$ be the Gram matrix of the limiting four-particle images $w^*_{Dr}=\sum_{p+j+k=D,\,j<k}S_r(D,D;p,j,k)\,v_p^*\wedge e_j\wedge e_k$, where $v_p^*$ are the limiting pair vectors \cite[(9)]{FockLaughlin}, as in \cite[(27)]{FockLaughlin}.
\end{definition}

These are \cite[(18), (20), (25), (27)]{FockLaughlin} with the index sets of Definition~\ref{def:ansatz}. The ranges are forced: an entry of the unaveraged three-body matrix \cite[(16)]{FockLaughlin} has total deficit $T=p+j=t_\rho+i\le T_3$, and an entry of the four-body matrix \cite[(17)]{FockLaughlin} has $T=p+j+k\le T_4$.

\begin{definition}\label{def:admissible}
A finite family of rows, together with numbers $\varepsilon_D\ge0$ for $1\le D\le T_4$, is \emph{admissible} if
\begin{align}
  e_z&<\tfrac32\,q_z^\infty\qquad\text{for every }z\in\{2,4,5,\dots,T_3\}, \label{eq:three}\\
  G_D\bigl(2\chi-E_D+\varepsilon_DI\bigr)G_D&\ge0\qquad\text{for every }1\le D\le T_4. \label{eq:four}
\end{align}
For such a family put $\bar\varepsilon_D=\max_{D\le D'\le T_4}\varepsilon_{D'}$, define the two allowance prices
\begin{equation}\label{eq:prices}
  \pi_{\mathrm B}(\varepsilon)=\sum_{D=1}^{T_4}n_D\,\varepsilon_D,\qquad
  \pi_{\mathrm S}(\varepsilon)=\sum_{D=1}^{T_4}\Bigl\lceil\frac D2\Bigr\rceil\bar\varepsilon_D ,
\end{equation}
and the two constants
\begin{equation}\label{eq:gammacert}
  \gamma^{\mathrm B}_{\mathrm{cert}}=1-\eta-\delta(T_3)-\pi_{\mathrm B}(\varepsilon),\qquad
  \gamma^{\mathrm S}_{\mathrm{cert}}=1-\eta-\delta(T_3)-\pi_{\mathrm S}(\varepsilon),
  \qquad \delta(Z)=\frac32\sum_{\substack{z>Z\\ z\ \mathrm{odd}}}(3z-1)\,2^{-z}\,c_z,
\end{equation}
where either $c_z=z+1$ \cite[Lemma~6.1]{FockLaughlin} or $c_z=3^z/\max_p\binom zp2^{z-p}$ (Lemma~\ref{lem:tail}).
\end{definition}

Since $\bar\varepsilon_D\le\sum_{D'\ge D}\varepsilon_{D'}$ and $n_{D'}=\sum_{D\le D'}\lceil D/2\rceil$, summation by parts gives $\pi_{\mathrm S}(\varepsilon)\le\pi_{\mathrm B}(\varepsilon)$, hence $\gamma^{\mathrm S}_{\mathrm{cert}}\ge\gamma^{\mathrm B}_{\mathrm{cert}}$, with equality when a single $\varepsilon_D$ is nonzero. The block-by-block price $\pi_{\mathrm B}$ is the Bessel price of \cite[Proposition~5.3]{FockLaughlin} applied to each block separately; the staircase price $\pi_{\mathrm S}$ uses the new Lemma~\ref{lem:descendant}. We write $\gamma_{\mathrm{cert}}$ for either constant when the distinction does not matter.

For the original data, \eqref{eq:three} and \eqref{eq:four} are \cite[(21), (28)]{FockLaughlin} with $\varepsilon_D=\varepsilon$.

\pagebreak[3]
\begin{theorem}\label{thm:framework}
Let a finite family of rows with allowances $\varepsilon_D$ be admissible.
\begin{enumerate}
\item[(a)] With the block-by-block price,
\[
  H_Q^2\ \ge\ d_a\int_{SU(2)}\sum_\rho F_\rho(g)^\dagger F_\rho(g)\,dg+\bigl(\gamma^{\mathrm B}_{\mathrm{cert}}-o(1)\bigr)H_Q\qquad\text{on }\Fock_Q,
\]
where $F_\rho(g)=U(g)F_\rho U(g)^\dagger$, $dg$ is normalized Haar measure, and $o(1)$ tends to zero as $Q\to\infty$, uniformly in particle number.
\item[(b)] The same holds with $\gamma^{\mathrm S}_{\mathrm{cert}}$ in place of $\gamma^{\mathrm B}_{\mathrm{cert}}$.
\end{enumerate}
In particular, for every $\gamma<\gamma^{\mathrm S}_{\mathrm{cert}}$ we have $H_Q^2\ge\gamma H_Q$ for all sufficiently large $Q$. With $c_z=z+1$ in $\delta(T_3)$, part~(a) uses only the arguments of \cite{FockLaughlin}, carried out on the enlarged index sets. Part~(b), and the choice of $c_z$ from Lemma~\ref{lem:tail}, use the new estimates of Sections~\ref{sec:tail} and~\ref{sec:allowance}.
\end{theorem}

\begin{proof}
We follow \cite[Sections~2--6]{FockLaughlin}, writing $K_Q=d_a\int\sum_\rho F_\rho(g)^\dagger F_\rho(g)\,dg$ and $A_2+A_3+A_4$ for its two-, three- and four-body parts. Take $Q>T_4+1$.

\emph{Normal-ordered square.} By \cite[Lemmas~2.1 and~3.2]{FockLaughlin},
\[
  H_Q^2=H_Q+\Lift_3\Bigl(W_3\sum_{z=0}^Qq_z(Q)R_{3,z}W_3^\dagger\Bigr)+\Lift_4(W_4W_4^\dagger).
\]
These identities do not involve the auxiliary operators.

\emph{Expansion of the squares.} The reordering $c_ic_k^\dagger=\delta_{ik}-c_k^\dagger c_i$ that gives \cite[(15)--(18)]{FockLaughlin} holds for each $F_\rho$ of the form~\eqref{eq:F}. Hence $A_2=\eta H_Q$, $A_3$ is the averaged lift of the matrix \cite[(16)]{FockLaughlin} on $V_Q\otimes U_Q$, and $A_4$ is the averaged lift of \cite[(17)]{FockLaughlin} on $V_Q\otimes\bigwedge^2U_Q$.

\emph{Three-body comparison.} The matrix \cite[(16)]{FockLaughlin} preserves total deficit $T\le T_3$, and $V_Q\otimes U_Q$ is multiplicity-free, so Schur averaging gives $A_3=\Lift_3\bigl(W_3\sum_{z\le T_3}(d_a/d_{3,z})e_z^QR_{3,z}W_3^\dagger\bigr)$, where $e_z^Q$ is the trace of the unaveraged matrix in block $z$. This trace involves finitely many coupling coefficients, so $e^Q_z\to e_z$ by \cite[Lemma~3.3]{FockLaughlin}. Since $d_a/d_{3,z}\to2/3$, $q_z(Q)\to q_z^\infty$, the inequalities~\eqref{eq:three} are strict and finite in number, and $W_3R_{3,z}=0$ for $z\in\{0,1,3\}$, we obtain for all large $Q$
\[
  A_3\le\Lift_3\Bigl(W_3\sum_{z\le T_3}q_z(Q)R_{3,z}W_3^\dagger\Bigr).
\]

\emph{Three-body tail.} By Lemma~\ref{lem:tailsum} below, $\Lift_3\bigl(W_3\sum_{T_3<z\le Q}q_z(Q)R_{3,z}W_3^\dagger\bigr)\ge-\delta_Q H_Q$ with $\delta_Q\to\delta(T_3)$. With $c_z=z+1$ this is \cite[Lemma~6.1]{FockLaughlin}, started at the first odd block above $T_3$ instead of $17$.

\emph{Four-body comparison.} The matrix \cite[(17)]{FockLaughlin} has entries only at total deficits $T\le T_4$, so its Schur average \cite[(24)]{FockLaughlin} lives on the spin blocks $1\le D\le T_4$; the target $\Lift_4(W_4W_4^\dagger)$ is positive on all other blocks. The transfer lemma \cite[Lemma~5.2]{FockLaughlin} applies with the local Fock space $\Fock_{\mathrm{loc}}=\bigwedge\mathbb C^{T_4+2}$ (modes $0,\dots,T_4+1$) in place of $\bigwedge\mathbb C^{25}$, and $h_Q=\sum_{p\le T_4}B_p^\dagger B_p$. Indeed, $w^Q_{Dr}$ involves $v_p$ with $p\le D$, hence modes up to $D+1$, and $e_j,e_k$ with $j,k\le D$. The proof of that lemma uses only that the local Fock space and the index set are fixed as $Q\to\infty$, and the non-strict limiting positivity \cite[(34)]{FockLaughlin}, which follows from~\eqref{eq:four} by \cite[Lemma~5.1]{FockLaughlin}. Proposition~\ref{prop:four} below then gives
\[
  A_4\le\Lift_4(W_4W_4^\dagger)+\bigl(\pi(\varepsilon)+o(1)\bigr)H_Q ,
\]
with $\pi=\pi_{\mathrm B}$ by part~(a) of that proposition, which is \cite[Proposition~5.3]{FockLaughlin} applied block by block, and with $\pi=\pi_{\mathrm S}$ by its part~(b).

\emph{Conclusion.} Adding the three comparisons and the tail to the normal-ordered square gives $H_Q^2\ge K_Q+(1-\eta-\delta_Q-\pi(\varepsilon)-o(1))H_Q$. All error terms involve finitely many fixed blocks and are operator inequalities on the whole Fock space, so they are uniform in particle number. This proves~(a) and~(b), and since $K_Q\ge0$ the last statement follows.
\end{proof}

\section{A product-twirl bound for the three-body tail}\label{sec:tail}

The three-body blocks with $z>T_3$ receive nothing from the auxiliary operators. For odd $z$ their coefficient $q_z(Q)$ is negative, and they must be paid for with pair energy. The original argument writes a normalized highest vector of block $z$ as $\xi_z=\sum_{p=0}^z s_pv_p\otimes e_{z-p}$ and applies Cauchy--Schwarz to $\|\sum_ps_pc_{z-p}B_p\psi\|^2$, at the cost of a factor $z+1$. Comparing the block with the rotation average of a single product vector is much cheaper.

\begin{lemma}\label{lem:tail}
Let $a\in V_Q$ and $b\in U_Q$ be unit vectors and $\theta=\langle a\otimes b,R_{3,z}(a\otimes b)\rangle>0$. Then
\[
  \Lift_3\bigl(W_3R_{3,z}W_3^\dagger\bigr)\le\frac{d_{3,z}}{d_a\,\theta}\,H_Q\qquad\text{on }\Fock_Q .
\]
In particular $\Lift_3(W_3R_{3,z}W_3^\dagger)\le(d_{3,z}/d_a)\,c_z(Q)\,H_Q$ with $c_z(Q)=1/\max_ps_p(Q)^2$, where $s_p(Q)$ are the coefficients of the normalized highest vector of block $z$. As $Q\to\infty$,
\begin{equation}\label{eq:cz}
  c_z(Q)\ \longrightarrow\ c_z=\frac{1}{\max_p\Bin(z,\tfrac13)(p)}=\frac{3^z}{\binom{z}{p_*}2^{z-p_*}},\qquad p_*=\Bigl\lfloor\frac{z+1}3\Bigr\rfloor,
\end{equation}
and $c_z(Q)\le z+1$ for every $Q$.
\end{lemma}

\begin{proof}
Write $x=a\otimes b$ and $x(g)=U(g)x$. Since $V_Q\otimes U_Q$ is multiplicity-free, Schur averaging gives
\[
  \int|x(g)\rangle\langle x(g)|\,dg=\sum_{z'}\frac{\langle x,R_{3,z'}x\rangle}{d_{3,z'}}R_{3,z'}\ \ge\ \frac{\theta}{d_{3,z}}R_{3,z},
\]
since every term is positive. The map $X\mapsto\Lift_3(W_3XW_3^\dagger)$ preserves order and intertwines rotations, so
\[
  \Lift_3(W_3R_{3,z}W_3^\dagger)\le\frac{d_{3,z}}{\theta}\int B\bigl(a(g)\wedge b(g)\bigr)^\dagger B\bigl(a(g)\wedge b(g)\bigr)\,dg .
\]
Now $B(a\wedge b)=B(b)B(a)$ and $\|B(b)\|\le1$, so the integrand is at most $B(a(g))^\dagger B(a(g))$, and $\int B(a(g))^\dagger B(a(g))\,dg=\Lift_2(\Pi_{V_Q})/d_a=H_Q/d_a$ by Schur's lemma in the irreducible $V_Q$. This proves the first claim.

For the second, take $a=v_p$ and $b=e_{z-p}$. The product $v_p\otimes e_{z-p}$ lies at total deficit $z$, and the deficit-$z$ weight space of block $z$ is spanned by its highest vector $\xi_z$, so $\theta=|\langle v_p\otimes e_{z-p},\xi_z\rangle|^2=s_p(Q)^2$. Since $\sum_ps_p(Q)^2=1$ over $z+1$ terms, $\max_ps_p(Q)^2\ge1/(z+1)$. By \cite[Lemma~3.3]{FockLaughlin} with $u^2=1/3$ and $T=z$, $s_p(Q)$ converges to the coefficient of $X^pY^{z-p}/\sqrt{p!(z-p)!}$ in $(uX-vY)^z/\sqrt{z!}$, whose square is $\binom zpu^{2p}v^{2(z-p)}=\Bin(z,\tfrac13)(p)$. The mode of $\Bin(z,\tfrac13)$ is $\lfloor(z+1)/3\rfloor$.
\end{proof}

\begin{lemma}\label{lem:tailsum}
For fixed $Z\ge2$ and all large $Q$,
\[
  \Lift_3\Bigl(W_3\sum_{Z<z\le Q}q_z(Q)R_{3,z}W_3^\dagger\Bigr)\ge-\delta_Q(Z)\,H_Q,\qquad
  \delta_Q(Z)=\sum_{\substack{Z<z\le Q\\ z\ \mathrm{odd}}}\frac{d_{3,z}}{d_a}\,|q_z(Q)|\,c_z(Q),
\]
and $\delta_Q(Z)\to\delta(Z)$ of~\eqref{eq:gammacert} with $c_z$ from~\eqref{eq:cz}. The same holds with $c_z(Q)$ and $c_z$ replaced by $z+1$.
\end{lemma}

\begin{proof}
For $2\le z\le Q$ the bracket in~\eqref{eq:qz} is at least $3z-1-(z+1)=2z-2>0$, so $q_z(Q)$ has the sign $(-1)^z$. Even blocks are discarded as positive, and Lemma~\ref{lem:tail} bounds each odd one. For convergence, $|q_z(Q)|\le3z(2/3)^z$ for $Q\ge4$, $d_{3,z}/d_a\le2$ and $c_z(Q)\le z+1$ give the dominating sequence $6z(z+1)(2/3)^z$ of \cite[Lemma~6.1]{FockLaughlin}, and each term converges by~\eqref{eq:qz} and~\eqref{eq:cz}.
\end{proof}

Table~\ref{tab:tail} compares the two tails. For the original index sets the product-twirl bound lowers the tail by $0.0107971$. For large index sets the tail is already small, but the new constant still reduces it by a factor of four to five.

\begin{table}[ht]
\centering
\renewcommand{\arraystretch}{1.15}
\small
\begin{tabular}{@{}lccc@{}}
\toprule
First uncovered odd block & $z=17$ ($T_3=15$) & $z=21$ ($T_3=19$) & $z=31$ ($T_3=29$)\\
\midrule
$c_z=z+1$ \cite{FockLaughlin} & $61/4096=0.0148926$ & $91/65536=0.0013885$ & $769/2^{28}=2.86\times10^{-6}$\\
$c_z$ of~\eqref{eq:cz} & $0.0040954$ & $0.0003422$ & $5.88\times10^{-7}$\\
\bottomrule
\end{tabular}
\caption{The three-body tail $\delta(T_3)$ of~\eqref{eq:gammacert}. The second row is summed exactly to $z=300$; the remainder is below $2^{-280}$.}
\label{tab:tail}
\end{table}

\begin{remark}\label{rem:pairparticle}
Lemma~\ref{lem:tail} is optimal among bounds obtained from rotation averages of product vectors: any separable operator dominating $R_{3,z}$ has normalized trace at least $1/\max_{a,b}\theta$, with $\theta$ as in Lemma~\ref{lem:tail}. Further gains are possible from the positivity of $O^\dagger O$ for operators $O=\sum_{p,m}\kappa_{pm}c_m^\dagger B_p$, which annihilate a pair and create one particle; we call them pair-to-particle operators. With $\bar\kappa=\sum_{p,m}\kappa_{pm}\,v_p\otimes e_m\in V_Q\otimes U_Q$ one finds $O^\dagger O=\Lift_2(\tr_U|\bar\kappa\rangle\langle\bar\kappa|)-\Lift_3\bigl(W_3(|\bar\kappa\rangle\langle\bar\kappa|)^{T_U}W_3^\dagger\bigr)$, with $T_U$ the partial transpose. The rotation averages of the partially transposed term can be negative on some blocks, so surplus in one block can pay for another. A linear program over such test vectors lowers the tail for the original rows from $0.0041$ to $0.0024$, giving $\gamma=0.05869$ in the setting of the third line of Table~\ref{tab:ledger}. We do not use this refinement, since for large index sets the tail is negligible.
\end{remark}

\section{Paying for the four-body allowance once}\label{sec:allowance}

The four-body condition~\eqref{eq:four} may need a small allowance $\varepsilon_DI$, which must be paid for in pair energy. In \cite[Proposition~5.3]{FockLaughlin} each block $D$ is paid for separately by Bessel's inequality over the $n_D$ triples $(p,j,k)$ at deficit $D$, for a total of $\sum_Dn_D\varepsilon$. The following lemma shows that a single vector at deficit $T$ already pays for all blocks $D\le T$.

For $T\ge1$ let $\mathcal B_T=\{b=(p,j,k):p+j+k=T,\ 0\le j<k\}$, so $|\mathcal B_T|=n_T$. For $x\in\R^{\mathcal B_T}$ put $h_x=\sum_{b\in \mathcal B_T}x_b\,v_p\otimes(e_j\wedge e_k)$ and $u_x=\Wh h_x$. For $D\le T$ let $S^Q(D,T)$ be the $|\mathcal R_D|\times n_T$ matrix of coefficients of the vectors $|D,T;r\rangle_Q$ in the basis $v_p\otimes(e_j\wedge e_k)$, $b\in \mathcal B_T$, and $w^Q_{Dr}=\Wh|D,D;r\rangle_Q$.

\begin{lemma}\label{lem:descendant}
Let $Q>T+1$ and $x\in\R^{\mathcal B_T}$, and put $y^D=S^Q(D,T)x$. Then
\[
  \int B(u_x(g))^\dagger B(u_x(g))\,dg=\sum_{D=1}^{T}\ \sum_{r,s\in\mathcal R_D}y^D_ry^D_s\int B(w^Q_{Dr}(g))^\dagger B(w^Q_{Ds}(g))\,dg .
\]
Moreover $S^Q(D,T)S^Q(D,T)^{\mathsf T}=I$ for every $D\le T$.
\end{lemma}

\begin{proof}
For $Q>T+1$ the vectors $v_p\otimes(e_j\wedge e_k)$, $b\in \mathcal B_T$, form an orthonormal basis of the weight-$T$ subspace of $V_Q\otimes\bigwedge^2U_Q$. So do the vectors $|D,T;r\rangle_Q$ with $1\le D\le T$ and $r\in\mathcal R_D$: they are normalized descendants in mutually orthogonal irreducible copies, and there are $\sum_{D\le T}\lceil D/2\rceil=n_T$ of them. Hence $S^Q(D,T)$ has orthonormal rows, and $h_x=\sum_{D\le T}\sum_ry^D_r|D,T;r\rangle_Q$. Schur averaging with equivalent copies \cite[Section~4]{FockLaughlin} sends $|h_x\rangle\langle h_x|$ to $\bigoplus_D d_D^{-1}I_{d_D}\otimes|y^D\rangle\langle y^D|$, which is also the average of $\sum_D\sum_{r,s}y^D_ry^D_s|D,D;r\rangle\langle D,D;s|$. Applying the order-preserving, rotation-intertwining map $X\mapsto\Lift_4(\Wh X\Wh^\dagger)$ gives the identity.
\end{proof}

\pagebreak[3]
\begin{proposition}\label{prop:four}
Suppose~\eqref{eq:four} holds, let $A_4$ be the four-body part of $K_Q$ as in the proof of Theorem~\ref{thm:framework}, and let $\bar\varepsilon_D=\max_{D\le D'\le T_4}\varepsilon_{D'}$. Then, uniformly on $\Fock_Q$:
\begin{enumerate}
\item[(a)] $A_4\le\Lift_4(W_4W_4^\dagger)+\bigl(\sum_{D=1}^{T_4}n_D\,\varepsilon_D+o(1)\bigr)H_Q$;
\item[(b)] $A_4\le\Lift_4(W_4W_4^\dagger)+\bigl(\sum_{D=1}^{T_4}\lceil D/2\rceil\,\bar\varepsilon_D+o(1)\bigr)H_Q$.
\end{enumerate}
Part~(a) is \cite[Proposition~5.3]{FockLaughlin} applied block by block, and does not use Lemma~\ref{lem:descendant}.
\end{proposition}

\begin{proof}
Let $\varepsilon'_D$ be $\varepsilon_D$ for~(a) and $\bar\varepsilon_D$ for~(b). Since $\varepsilon'_D\ge\varepsilon_D$, \eqref{eq:four} also holds with $\varepsilon'_D$. As in \cite[Proposition~5.3]{FockLaughlin}, the generalized transfer lemma applied to $C^Q_D=(d_D/d_a)\chi-E^Q_D+\varepsilon'_DI$, conjugated by rotations and averaged, gives
\[
  \Lift_4(W_4W_4^\dagger)-A_4+\sum_{D=1}^{T_4}\varepsilon'_D\,X_D\ \ge\ -o(1)H_Q,\qquad
  X_D=d_a\int\sum_{r\in\mathcal R_D}B(w^Q_{Dr}(g))^\dagger B(w^Q_{Dr}(g))\,dg\ \ge0 .
\]
It remains to bound the allowance.

\emph{Part (a).} This is the Bessel argument of \cite[Proposition~5.3]{FockLaughlin}, applied to each block. Write $|D,D;r\rangle_Q=\sum_{b\in \mathcal B_D}S^Q_{rb}\,v_p\otimes(e_j\wedge e_k)$. The vectors on both sides are orthonormal, so $S^Q(S^Q)^{\mathsf T}=I$, and for every $\psi$
\[
  \sum_{r\in\mathcal R_D}\|B(w^Q_{Dr})\psi\|^2\le\sum_{b\in \mathcal B_D}\|c_kc_jB_p\psi\|^2\le\sum_{b\in \mathcal B_D}\|B_p\psi\|^2 .
\]
Conjugating by rotations and averaging gives $X_D\le n_D\,H_Q$, because $d_a\int B_p(g)^\dagger B_p(g)\,dg=H_Q$ and $|\mathcal B_D|=n_D$. Hence $\sum_D\varepsilon_DX_D\le\pi_{\mathrm B}(\varepsilon)H_Q$.

\emph{Part (b).} Put $\omega_T=\bar\varepsilon_T-\bar\varepsilon_{T+1}\ge0$, with $\bar\varepsilon_{T_4+1}=0$, so that $\bar\varepsilon_D=\sum_{T\ge D}\omega_T$. Apply Lemma~\ref{lem:descendant} to $x=\mathbf e_b$ for each $b\in \mathcal B_T$ and sum over $b$. Since $\sum_b(S^Q\mathbf e_b)(S^Q\mathbf e_b)^{\mathsf T}=S^Q(S^Q)^{\mathsf T}=I$, this gives
\[
  \sum_{D\le T}X_D=d_a\sum_{b\in \mathcal B_T}\int B(v_p(g)\wedge e_j(g)\wedge e_k(g))^\dagger B(v_p(g)\wedge e_j(g)\wedge e_k(g))\,dg\le n_T\,H_Q,
\]
because $B(v_p\wedge e_j\wedge e_k)=c_kc_jB_p$, $\|c_kc_j\|\le1$ and $d_a\int B_p(g)^\dagger B_p(g)\,dg=H_Q$. Therefore
\[
  \sum_D\bar\varepsilon_DX_D=\sum_T\omega_T\sum_{D\le T}X_D\le\sum_T\omega_Tn_T\,H_Q=\sum_D\bar\varepsilon_D(n_D-n_{D-1})\,H_Q=\sum_D\Bigl\lceil\frac D2\Bigr\rceil\bar\varepsilon_D\,H_Q,
\]
by summation by parts.
\end{proof}

\begin{corollary}\label{cor:144}
For the certificate of \cite{FockLaughlin}, with $\varepsilon_D=\varepsilon=3\cdot10^{-6}$ for $D\le23$, Proposition~\ref{prop:four}(b) prices the allowance at $n_{23}\varepsilon=144\varepsilon$ instead of $1222\varepsilon$. Together with Lemma~\ref{lem:tailsum}, the original seven rows give every $\gamma$ below
\[
  1-\frac{93527408868499}{10^{14}}-\delta(15)-144\cdot\frac{3}{10^6}=0.0601984667\ldots,
\]
where $\delta(15)=0.0040954445\ldots$ is computed with~\eqref{eq:cz}.
\end{corollary}

\begin{remark}
Within this pricing scheme the identity allowance costs at least $\varepsilon\sum_{D\le23}|\mathcal R_D|=144\varepsilon$, so Corollary~\ref{cor:144} is optimal for it. Allowances adapted to the few negative directions are cheaper still: for the original rows, sixteen rank-one allowances, each built from a vector $u_x$ supported on two triples with a common label $p$ and priced at $\|x\|_2^2$, cost $8.9\times10^{-6}$ in total and give $\gamma=0.0498244$, with the original tail $61/4096$. Allowances of this form are not covered by Theorem~\ref{thm:framework}, and re-optimized auxiliary operators need no allowance in most blocks, so we do not pursue this further.
\end{remark}

\section{The semidefinite program}\label{sec:sdp}

The constant of Theorem~\ref{thm:framework} depends on the choice of auxiliary operators. This section explains how they are chosen: as the optimum of a semidefinite program (SDP). Its numerical output is then converted into an exact certificate (Section~\ref{sec:exact}). The program plays the role of a proof search. Nothing in the proof depends on it, since every certificate is verified exactly.

\subsection{Variables and objective}

For each $t\le\tbar$ collect the rows with $t_\rho=t$ into the Gram matrix
\[
  X_t=\sum_{\rho:\,t_\rho=t}f_\rho f_\rho^{\mathsf T},\qquad f_\rho=\bigl(\lambda_\rho,(\alpha^\rho_{pj})_{(p,j)\in S_t}\bigr),
\]
a positive semidefinite matrix of size $1+|S_t|$ with index set $\{0\}\cup S_t$. The index $0$ refers to the pair annihilator $B_t$ and $(p,j)$ to the hopping term $c^\dagger_{p+j-t}c_jB_p$.

Only these Gram matrices matter. The averaged square $d_a\int\sum_\rho F_\rho(g)^\dagger F_\rho(g)\,dg$ is quadratic in the rows, so it depends on them only through $\sum_\rho f_\rho f_\rho^{\mathsf T}$. Rows with different shifts never interfere. A row with shift $t$ lowers the total orbital label by $t+1$, so cross terms between rows of different shifts change the $L_z$ quantum number and vanish under rotation averaging. Conversely, every family of positive semidefinite $X_t$ arises from rows, by spectral decomposition. The variables are therefore the matrices $X_t\ge0$, $0\le t\le\tbar$, together with the allowances $\varepsilon_D\ge0$. The objective is to minimize
\begin{equation}\label{eq:objective}
  \eta+\pi(\varepsilon)=\sum_t(X_t)_{00}+\sum_D\pi_D\,\varepsilon_D ,
\end{equation}
which by~\eqref{eq:gammacert} is the same as maximizing $\gamma^{\mathrm B}_{\mathrm{cert}}=1-\eta-\delta(T_3)-\pi_{\mathrm B}(\varepsilon)$ when $\pi_D=n_D$, because the tail $\delta(T_3)$ is fixed by the index bounds. In our computations $\pi_D=n_D$, the price of \cite[Proposition~5.3]{FockLaughlin} applied block by block.

\subsection{The constraints}

All certificate data are linear in $X=(X_t)_t$. For the three-body traces~\eqref{eq:ez}, let $\phi^t_{z,T}\in\R^{\{0\}\cup S_t}$ be the vector with entries $s_{z,T,t}$ at index $0$ and $s_{z,T,p}$ at each $(p,j)$ with $p+j=T$, and zero elsewhere. Then
\begin{equation}\label{eq:ezlinear}
  e_z=\sum_t\langle\Phi^t_z,X_t\rangle,\qquad
  \Phi^t_z=\sum_{T\ge\max(z,t)}\Bigl(\phi^t_{z,T}\,(\phi^t_{z,T})^{\mathsf T}-s_{z,T,t}^2\,\mathbf e_0\mathbf e_0^{\mathsf T}\Bigr).
\end{equation}
For the four-body matrices~\eqref{eq:ED}, $E_D=-\sum_t\Psi^t_D(X_t)$, where $\Psi^t_D$ maps a matrix on $\{0\}\cup S_t$ to a symmetric matrix on the copy labels $\mathcal R_D$:
\[
  \Psi^t_D(Y)_{rs}=\sum_{(p,j),(q,l)\in S_t,\ T\ge D}Y_{(p,j),(q,l)}\,S_r(D,T;p,j,k)\,S_s(D,T;q,l,i),
\]
with $i=p+j-t$, $k=q+l-t$ and $T=p+j+k$. The index $0$ does not enter: the pair annihilators contribute nothing to the four-body term. The program is
\begin{equation}\label{eq:sdp}
\begin{aligned}
  \text{minimize}\quad&\sum_t(X_t)_{00}+\sum_D\pi_D\varepsilon_D\\
  \text{subject to}\quad&X_t\ge0,\qquad \varepsilon_D\ge0,\\
  &\textstyle\sum_t\langle\Phi^t_z,X_t\rangle\le\tfrac32q_z^\infty\qquad(z\in\{2,4,5,\dots,T_3\}),\\
  &V_D^{\mathsf T}\Bigl(2\chi+\textstyle\sum_t\Psi^t_D(X_t)+\varepsilon_DI\Bigr)V_D\ge0\qquad(1\le D\le T_4),
\end{aligned}
\end{equation}
where $V_D$ is an orthonormal basis of $\operatorname{ran}G_D$. Each constraint has a direct meaning.
\begin{itemize}
\item $X_t\ge0$ makes the comparison operator a sum of squares, hence positive on every Fock space.
\item The three-body constraint for block $z$ is the limit, in that block, of the comparison
\[
  A_3\le\Lift_3\Bigl(W_3\sum_zq_z(Q)R_{3,z}W_3^\dagger\Bigr).
\] The square of $H_Q$ contributes $q_z(Q)$ per unit of $R_{3,z}$, and the averaged three-body part of $K_Q$ contributes $(d_a/d_{3,z})e^Q_z\to\frac23e_z$. The blocks $z\in\{0,1,3\}$ are absent because wedging annihilates them.
\item The four-body constraint for block $D$ compares, on the copy space of spin $2Q-1-D$, the target $\Lift_4(W_4W_4^\dagger)$ with the four-body part of $K_Q$. The target supplies only the copy $r=1$, with weight $d_D/d_a\to2$, which is the term $2\chi$. The restriction to $\operatorname{ran}G_D$ expresses that only the wedged four-particle images matter \cite[Lemma~5.1]{FockLaughlin}, and these images may be linearly dependent.
\end{itemize}
After rotation averaging, the operator inequality on all Fock spaces has become one scalar inequality per three-body block and one small matrix inequality per four-body block. For the certificate of Appendix~\ref{app:main}, with $(\pbar,\jbar,\ibar,\tbar)=(17,18,18,17)$, the program has $18$ Gram blocks of size up to $262$, $33$ three-body inequalities, and $53$ four-body matrix inequalities of size at most $27$. That is $2139$ scalar constraints after removal of common null spaces.

\subsection{Why the program is effective}

\paragraph{Convexity.} The objective and all constraints are linear in $(X_t,\varepsilon_D)$, and the cones are positive semidefinite. The program is therefore convex: every local optimum is global, and an interior-point method finds it to high accuracy \cite{VandenbergheBoyd1996}. In our runs the primal and dual objectives agree to between $10^{-8}$ and $10^{-6}$. So for each choice of index sets we know, to about six digits, the largest value of $\gamma^{\mathrm B}_{\mathrm{cert}}$, with $c_z=z+1$ and $\pi_D=n_D$, that operators of the form~\eqref{eq:F} can give. The seven rows of \cite{FockLaughlin} are feasible but not optimal. Within the same index sets the optimum is $0.0636648$ rather than $\GammaOld$, and needs no allowance.

\paragraph{Symmetry reduction.} The inequality $H_Q^2\ge\gamma H_Q$ concerns infinitely many Fock spaces, one for each flux and particle number. Rotation averaging and Schur's lemma replace it by finitely many comparisons on spin blocks, all made before the lift to many particles. This is what makes the certificate uniform in particle number, and it keeps the program small. The $L_z$ grading makes it block diagonal in $t$.

\paragraph{The trade-off is explicit.} By~\eqref{eq:ezlinear}, a single row contributes
\[
  e_z(f)=\sum_{T}\langle\phi^t_{z,T},f\rangle^2-\lambda^2\sum_Ts_{z,T,t}^2\ \ge\ -\lambda^2\sum_Ts_{z,T,t}^2 .
\]
A negative trace in block $z$, as required for odd $z$ by $q_z^\infty<0$, can only be bought with pair energy $\lambda^2$. The hopping coefficients must make the projections $\langle\phi^t_{z,T},f\rangle$ small in every block where the trace has to be negative. They cannot be chosen freely, because their squares feed the four-body matrices $E_D$, which must stay below the single copy $2\chi$ supplied by the target. The program solves this balance exactly, and the optimum is far from obvious. In the main certificate most of the pair energy sits in the rows with $t=0$ and $t=1$, which contribute about $0.30$ and $0.43$ to $\eta\approx0.78$. The hopping parts of these rows then reach high deficits.

\paragraph{Why enlarging helps.} Larger bounds $\pbar,\jbar$ cover more three-body blocks exactly, so the tail $\delta(T_3)$ becomes negligible ($\delta(35)=271/2^{32}$). A larger bound $\ibar$ on the created mode admits hopping terms with $p+j$ well above $t$. This lets the cheap low-$t$ rows act on many blocks, and gives the four-body constraints more directions in which to cancel. Empirically $\ibar$ is the most effective single parameter (Table~\ref{tab:trend}).

\paragraph{Low rank.} The optimal $X_t$ have low numerical rank. In all our certificates most shifts carry at most five rows; only shifts at or next to $\tbar$ sometimes need more (fifteen at $t=\tbar$ in the main certificate). The certificates therefore consist of a few dozen rows. This keeps the exact verification of Section~\ref{sec:exact} cheap and the rounding stable.

\subsection{The dual program}

The dual of~\eqref{eq:sdp} has multipliers $\mu_z\ge0$ for the three-body inequalities and positive semidefinite $\Lambda_D$ on $\operatorname{ran}G_D$ for the four-body ones. Writing $\Omega_D=V_D\Lambda_DV_D^{\mathsf T}$, it reads
\begin{equation}\label{eq:dual}
\begin{aligned}
  \text{maximize}\quad&-\tfrac32\sum_zq_z^\infty\mu_z-2\sum_D(\Omega_D)_{11}\\
  \text{subject to}\quad&\lambda^2+\sum_z\mu_z\,e_z(f)+\sum_D\langle\Omega_D,E_D(f)\rangle\ \ge\ 0\qquad\text{for every row }f,\\
  &\tr\Lambda_D\le\pi_D,\qquad\mu_z\ge0,\qquad\Lambda_D\ge0 .
\end{aligned}
\end{equation}
Here $e_z(f)$ and $E_D(f)$ are the data of a single row $f$ of any shift. The constraint is a positivity condition. The multipliers define a linear functional on the two-, three- and four-body parts of the averaged squares, normalized by its value on the pair energy, and this functional must be nonnegative on $F^\dagger F$ for every operator $F$ of the form~\eqref{eq:F}. In other words, the dual searches for rotation-invariant pseudo-moments that are positive on the squares of the ansatz. This is a bootstrap problem in the sense of \cite{Gao2025}. Weak duality bounds the value of~\eqref{eq:sdp} from below by the value of~\eqref{eq:dual}, so every dual-feasible point bounds $\gamma^{\mathrm B}_{\mathrm{cert}}$ from above for the given index sets.

We used the dual in three ways, always numerically and never as part of a proof.
\begin{itemize}
\item \emph{Optimality.} The duality gap certifies, up to floating-point accuracy, that the optima in Table~\ref{tab:trend} are the largest values of $\gamma^{\mathrm B}_{\mathrm{cert}}$ (with $c_z=z+1$ and $\pi_D=n_D$) for each family.
\item \emph{Diagnostics.} The multipliers show which comparisons bind. For the diagonal index sets the largest $\mu_z$ sits at the last covered three-body block $z=T_3$; when $\ibar>\jbar$ it sits at $z=5$. For the large diagonal index sets only one or two allowances reach their price ($\tr\Lambda_D=\pi_D$) at the optimum: $D=2\jbar+1$, and for $n=18$ also $D=2\jbar+3$. This is why the allowance and the tail end up negligible and $\eta$ is the quantity that limits $\gamma$.
\item \emph{Pricing new operators.} A new family of auxiliary operators can raise the optimum only if some square in that family violates the dual constraint at the current multipliers. Consider the pair-to-particle operators $O_\kappa=\sum\kappa_{pm}c_m^\dagger B_p$ of Remark~\ref{rem:pairparticle} with a fixed $L_z$ shift $\tau$, normalized by $\|\kappa\|=1$ (pair energy one), and let $-\langle\kappa,\widetilde\Phi^\tau_z\kappa\rangle$ be the limiting three-body trace of $O_\kappa^\dagger O_\kappa$ in block $z$. The test is then $\lambda_{\max}(\sum_z\mu_z\widetilde\Phi^\tau_z)>1$. At the optimum for $(14,15,15,14)$ the largest value over the covered blocks is $0.62$. Such operators therefore cannot improve the certificate beyond removing the already negligible tail, as a direct computation confirms.
\end{itemize}

\subsection{From numerical optima to exact certificates}\label{sec:exact}

We solved~\eqref{eq:sdp} with our own dense primal--dual interior-point method, using the HKM search direction \cite{HKM1996,KSH1997,Monteiro1997} and Mehrotra's predictor--corrector scheme \cite{Mehrotra1992}, after removing common null spaces of the constraint matrices. The code is distributed with the ancillary files (directory \texttt{sdp}). In these runs each three-body inequality carried a slack whose violation was priced at $c_z=z+1$; at every optimum the priced violation was below $2\times10^{-8}$, and for $(7,8,8,7)$ and $(9,10,10,9)$ solving~\eqref{eq:sdp} with hard constraints gives the same optima to all digits shown in Table~\ref{tab:trend}.

A numerical optimum of~\eqref{eq:sdp} lies on many tight constraints, so it cannot be rounded directly. We re-solve with the constraints tightened by margins $\tau_3,\tau_4$ of order $10^{-9}$ to $10^{-7}$, adding a small penalty on the rank of each $X_t$. We then extract low-rank rows from the spectral decomposition of each $X_t$, keeping the eigenvectors whose eigenvalue exceeds $10^{-3}$ times the largest. The rows are rounded in the normalization of \cite[(39)]{FockLaughlin}:
\[
  \alpha^\rho_{pj}=\frac{a^\rho_{pj}}{P}\sqrt{\frac{2^{p-t}\,t!\,(p+j-t)!}{p!\,j!}},\qquad\lambda_\rho=\frac{\ell_\rho}{P},
\]
with integers $a^\rho_{pj},\ell_\rho$ and $P=10^{10}$ to $10^{15}$. The pair coefficients $\ell_\rho$ do not enter $E_D$; after rounding they are re-optimized to minimize $\eta$ subject to three-body margins of order $10^{-10}$ to $10^{-8}$. Finally the smallest admissible allowances $\varepsilon_D$ on the grid $1/P$ are found by exact bisection. All radicals cancel in $e_z$, and those in $E_D$ and $G_D$ are removed by the diagonal congruence of \cite[(46)]{FockLaughlin}. Admissibility is then checked in exact rational arithmetic:
\begin{itemize}
\item every three-body margin $\tfrac32q_z^\infty-e_z$ is computed as a rational number and checked to be positive;
\item every four-body condition~\eqref{eq:four} is checked through the congruent matrix $M_D=Z_DB_DZ_D$ of \cite[(46)]{FockLaughlin}, both by an exact symmetric $LDL^{\mathsf T}$ elimination and by the signs of the characteristic coefficients, computed with the Leverrier--Faddeev recurrence \cite{Hou1998} (the positivity test of \cite[Theorem~5]{PeyrlParrilo2008}).
\end{itemize}
Gram matrices with kernels cause no difficulty: only the positivity of $M_D$ is used, never its invertibility. Several of our certificates have $M_D$ with exact zero eigenvalues inside $\operatorname{ran}G_D$. These are confirmed by the exact checks and are harmless, since~\eqref{eq:four} is non-strict.

Each certificate reported below was verified by at least two exact implementations of these formulas, written independently of each other. One of them follows \cite[(40)--(47)]{FockLaughlin} literally, and for the larger certificates another, written separately from the search code, was built directly from the definitions \cite[(13), (20), (23), (25), (27)]{FockLaughlin}. All of them reproduce the constant $4616733319001/10^{14}$ of \cite{FockLaughlin} exactly from its seven rows. In negative tests, perturbing all integers $a^\rho_{pj}$ of a certificate by a relative $10^{-6}$, or dropping any row, makes the check fail; a single integer, by contrast, can often be changed by this amount within the margins. As an independent check of the limiting formulas on the enlarged index ranges, we built the coupled bases at finite flux from $SU(2)$ ladder operators. For $Q$ between $10^3$ and $10^5$ and a sample of blocks with $z\le35$ and $D\le53$, the finite-flux traces $e^Q_z$ and matrices $E^Q_D$, $G^Q_D$ converge to the limiting data at rate $O(1/Q)$.

\subsection{Results}

Table~\ref{tab:trend} lists the optimal values of~\eqref{eq:sdp}, expressed as the constants $\gamma^{\mathrm B}_{\mathrm{cert}}$, for the original tail constant $c_z=z+1$ and the price $\pi_D=n_D$ of \cite{FockLaughlin}. Figure~\ref{fig:trend} plots them along the diagonal $(\pbar,\jbar,\ibar,\tbar)=(n-1,n,n,n-1)$, together with the certified constants. Along the diagonal the optimum rises from $0.0637$ at the original size $n=8$ to $0.2190$ at $n=18$, and the increments shrink steadily. At fixed $(\pbar,\jbar)$, increasing the bound $\ibar$ on the created mode is the most effective single change: at $(7,8)$ it doubles the constant before saturating near $\ibar=15$. For $11\le n\le17$ exactly one four-body block, $D=2\jbar+1$, is limited by its allowance price; for $n=18$ two are, $D=37$ and $D=39$.

\begin{table}[ht]
\centering
\renewcommand{\arraystretch}{1.12}
\begin{tabular}{@{}cccc|cc@{}}
\toprule
$n$ & $(\pbar,\jbar,\ibar,\tbar)$ & SDP optimum & certified & $(\pbar,\jbar,\ibar,\tbar)$ & SDP optimum\\
\midrule
 8 & $(7,8,8,7)$     & $0.0636648$ & $0.0636643$ & $(7,8,10,7)$   & $0.1089344$\\
 9 & $(8,9,9,8)$     & $0.1208114$ & $0.1208108$ & $(7,8,12,7)$   & $0.1252527$\\
10 & $(9,10,10,9)$   & $0.1578160$ & $0.1578159$ & $(7,8,15,7)$   & $0.1280643$\\
11 & $(10,11,11,10)$ & $0.1848127$ & $0.1848125$ & $(8,9,17,8)$   & $0.1691831$\\
12 & $(11,12,12,11)$ & $0.2000745$ & $0.2000742$ & $(9,10,13,9)$  & $0.1888610$\\
13 & $(12,13,13,12)$ & $0.2070467$ & $0.2070464$ & $(9,10,19,9)$  & $0.1906601$\\
14 & $(13,14,14,13)$ & $0.2109270$ & $0.2109198$ & $(10,11,15,10)$ & $0.2019389$\\
15 & $(14,15,15,14)$ & $0.2138666$ & $0.2138574$ & $(12,13,16,12)$ & $0.2114576$\\
16 & $(15,16,16,15)$ & $0.2162413$ & $0.2162005$ & $(14,15,17,14)$ & $0.2157530$\\
17 & $(16,17,17,16)$ & $0.2179212$ & $0.2178926$ & $(15,16,18,15)$ & $0.2175081$\\
18 & $(17,18,18,17)$ & $0.2190432$ & $\mathbf{0.2190037}$ & $(16,17,19,16)$ & $0.2187575$\\
\bottomrule
\end{tabular}
\caption{Optimal constants $\gamma^{\mathrm B}_{\mathrm{cert}}$ of~\eqref{eq:sdp} in floating point, rounded to seven digits (the last digit is uncertain for $n=18$ and for $(9,10,19,9)$), and exactly certified constants, truncated. Certified constants are also available for $(9,10,13,9)$, namely $0.1882970$ with $\varepsilon_D=0$ for all $D$, and for $(15,16,18,15)$, namely $0.2174830$.}
\label{tab:trend}
\end{table}

\begin{figure}[ht]
\centering
\begin{tikzpicture}[x=0.9cm,y=22cm]
  \draw[->] (7.6,0.04) -- (19.4,0.04) node[right,font=\small] {$n$};
  \draw[->] (7.6,0.04) -- (7.6,0.245) node[above,font=\small] {$\gamma$};
  \foreach \x in {8,...,18} \draw (\x,0.04) -- (\x,0.036) node[below,font=\scriptsize] {\x};
  \foreach \y in {0.05,0.10,0.15,0.20} \draw (7.6,\y) -- (7.5,\y) node[left,font=\scriptsize] {\y};
  \draw[dashed,gray] (7.6,0.2) -- (19.2,0.2);
  \draw[red!70!black] plot[only marks,mark=square*,mark size=1.6pt] coordinates {(8,0.0636643) (9,0.1208108) (10,0.1578159) (11,0.1848125) (12,0.2000742) (13,0.2070464) (14,0.2109198) (15,0.2138574) (16,0.2162005) (17,0.2178926) (18,0.2190037)};
  \draw[thick,blue!60!black] plot[mark=*,mark size=0.9pt] coordinates {(8,0.0636648) (9,0.1208114) (10,0.1578160) (11,0.1848127) (12,0.2000745) (13,0.2070467) (14,0.2109270) (15,0.2138666) (16,0.2162413) (17,0.2179212) (18,0.2190432)};
  \node[font=\scriptsize,anchor=west] at (14.4,0.12) {\textcolor{blue!60!black}{$\bullet$ SDP optimum}};
  \node[font=\scriptsize,anchor=west] at (14.4,0.105) {\textcolor{red!70!black}{$\blacksquare$ certified}};
  \draw[gray,thick] (7.6,0.0461673) -- (8.1,0.0461673);
  \node[font=\scriptsize,anchor=west,gray] at (8.15,0.0461673) {original constant $\GammaOld$};
\end{tikzpicture}
\caption{SDP optima (circles) and exactly certified constants (squares) of Table~\ref{tab:trend} along the diagonal index bounds $(\pbar,\jbar,\ibar,\tbar)=(n-1,n,n,n-1)$. The dashed line is $1/5$, the constant in Corollary~\ref{cor:laughlin}; the short gray mark is the original constant $\gamma_*^{\mathrm{old}}$.}
\label{fig:trend}
\end{figure}
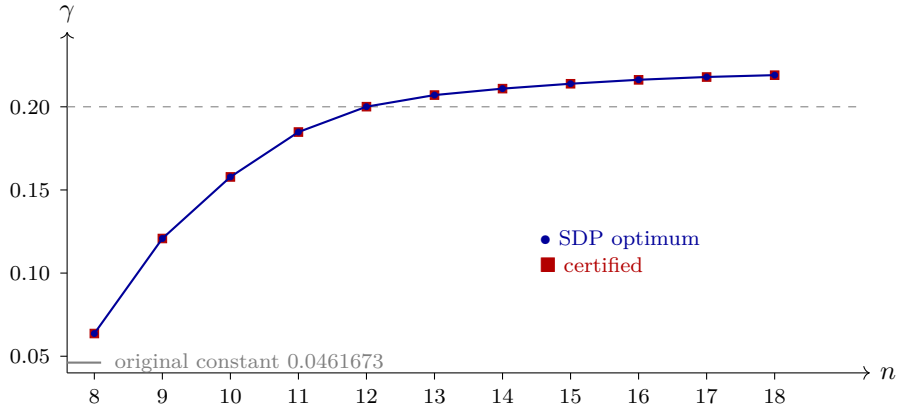

The best certificate, described in Appendix~\ref{app:main}, has $(\pbar,\jbar,\ibar,\tbar)=(17,18,18,17)$, so $T_3=35$ and $T_4=53$. It consists of $64$ rows over $P=10^{13}$, with all three-body margins positive (the smallest is $1.99999\times10^{-8}$, at $z=5$). It uses seven small allowances, at $D=13,14,15,17,20,37,39$. Priced block by block as in \cite[Proposition~5.3]{FockLaughlin}, they cost
\[
  \sum_Dn_D\varepsilon_D=\frac{59808781}{25\cdot10^{11}}=2.39\times10^{-5},
\]
and each is tight: lowering any $\varepsilon_D$ by $10^{-13}$ makes its block fail. The constant is
\begin{equation}\label{eq:gammastar}
\begin{gathered}
  \gamma^{\mathrm B}_{\mathrm{cert}}=1-\eta-\frac{271}{2^{32}}-\sum_Dn_D\varepsilon_D=\frac{\GammaNum}{\GammaDen},\\
  \eta=\frac{39048611444208798479924001}{5\cdot10^{25}}=0.78097222\ldots .
\end{gathered}
\end{equation}
Here $271/2^{32}=\delta(35)$ with $c_z=z+1$, so~\eqref{eq:gammastar} is the constant $\gamma^{\mathrm B}_{\mathrm{cert}}$ of Definition~\ref{def:admissible} computed with the original tail constant and the block-by-block allowance price, exactly as in \cite{FockLaughlin}. The previous index bounds, $(16,17,17,16)$, give the certified constant $0.2178926$, with a single allowance at $D=35$.

\begin{remark}\label{rem:outlook}
Along the diagonal of Table~\ref{tab:trend} the increments of the SDP optimum are $0.057$, $0.037$, $0.027$, $0.015$, $0.0070$, $0.0039$, $0.0029$, $0.0024$, $0.0017$ and $0.0011$. Over the last four steps they shrink by a factor of about $0.7$, and extrapolating them geometrically suggests a supremum of about $0.222$, in the range $0.221$ to $0.224$, for $\gamma^{\mathrm B}_{\mathrm{cert}}$ with auxiliary operators of the form~\eqref{eq:F}; this is a heuristic extrapolation, not a bound. For comparison, every positive planar eigenvalue is an upper bound on the best possible constant in Theorem~\ref{thm:main}, since the spherical spectrum converges to the planar one \cite[Section~7]{FockLaughlin}. Our exact diagonalization, which is not part of the proof, gives the lowest positive planar eigenvalues $9/16$ for $N=3$ and $0.4568$, $0.4192$, $0.4130$, $0.4099$ and $0.4075$ for $N=4,\dots,8$, over the angular momenta searched; on the sphere, the neutral gap at $Q=3(N-1)$ for $N\le12$, extrapolated in $1/N$, is about $0.40$ to $0.42$. The true gap is therefore presumably near $0.4$.
\end{remark}

\section{Proofs of the results}\label{sec:proofs}

\begin{proof}[Proof of Theorem~\ref{thm:main}]
The certificate of Appendix~\ref{app:main} is admissible in the sense of Definition~\ref{def:admissible}, by the exact checks described there. Its constant $\gamma^{\mathrm B}_{\mathrm{cert}}$, computed with $c_z=z+1$, is $\gamma_*$ of~\eqref{eq:gammastar}. Theorem~\ref{thm:framework}(a) gives the claim.
\end{proof}

This proof uses only the arguments of \cite{FockLaughlin}, carried out on the enlarged index sets, together with the exact certificate. It does not use Lemmas~\ref{lem:tail}, \ref{lem:tailsum} (beyond the case $c_z=z+1$, which is \cite[Lemma~6.1]{FockLaughlin}) or~\ref{lem:descendant}. The new estimates improve the same certificate slightly. Using $c_z$ of~\eqref{eq:cz} instead of $z+1$ raises the constant by $5\times10^{-8}$. Keeping $c_z=z+1$ but pricing the allowances with the staircase price $\pi_{\mathrm S}$ of Proposition~\ref{prop:four}(b), as in Theorem~\ref{thm:framework}(b), raises it to
\[
  \gamma^{\mathrm S}_{\mathrm{cert}}=\frac{1401708552455165441298243497}{64\cdot10^{26}}=0.2190169613\ldots .
\]
With both improvements the constant is $0.2190170124\ldots$. We state Theorem~\ref{thm:main} with the more conservative value~\eqref{eq:gammastar}.

\begin{proof}[Proof of Corollary~\ref{cor:laughlin}]
Apply Theorem~\ref{thm:main} with $\gamma=1/5<\gamma_*$. On each particle sector the spectral theorem gives $H_{N,Q}\ge\frac15(I-P_{\ker H_{N,Q}})$ once $Q\ge Q_{1/5}$. At $Q=3(N-1)$ the kernel is the line spanned by $\Psi_{\mathrm L,N}$ \cite[Lemma~6.2]{FockLaughlin}. Choose $N_0$ with $3(N_0-1)\ge Q_{1/5}$.
\end{proof}

\begin{proof}[Proof of Corollary~\ref{cor:planar}]
This is \cite[Section~7]{FockLaughlin} verbatim with the new constant. Each block of fixed particle number and total degree is finite-dimensional and invariant. The spherical restrictions converge in norm to the planar ones, and~\eqref{eq:main} passes to the limit for each $\gamma<\gamma_*$, then to $\gamma=\gamma_*$ by closedness of the positive cone.
\end{proof}

\begin{proof}[Proof of Corollary~\ref{cor:charge}]
For $N\ge2$ and $Q_N=3(N-1)$, \cite[Corollary~6.3]{FockLaughlin}, a specialization of the recursive spectral relations of \cite{LNWY2026}, gives $\Delta_N=E_0(N+1;Q_N)\ge\frac{N}{N-1}\,\mathrm{Gap}_N$. Here $\mathrm{Gap}_N$ is the smallest positive eigenvalue of $H_{N,Q_N}$, and Corollary~\ref{cor:laughlin} gives $\mathrm{Gap}_N\ge1/5$ for $N\ge N_0$.
\end{proof}

For completeness, we note that the comparison in \cite[Corollary~6.3]{FockLaughlin} follows directly from the identity $\sum_jc_j^\dagger H_{N,Q}c_j=(N-1)H_{N+1,Q}$ on $\bigwedge^{N+1}U_Q$. It holds because $B_p$ commutes with every $c_j$ and $B_p^\dagger$ with every $c_j^\dagger$, so that $\sum_jc_j^\dagger B_p^\dagger B_pc_j=B_p^\dagger\hat NB_p$ with $\hat N=N-1$ on the range of $B_p$. One combines it with $H_{N,Q_N}\ge\mathrm{Gap}_N(I-P_{\mathrm L,N})$ and $\|\sum_jc_j^\dagger P_{\mathrm L,N}c_j\|\le\operatorname{rank}P_{\mathrm L,N}=1$.

\appendix

\section{The certificates}\label{app:cert}

\subsection{Format and arithmetic}

A certificate lists the index bounds $(\pbar,\jbar,\ibar,\tbar)$, the denominator $P$, the allowances $\varepsilon_D$, and the rows. Each row gives $t$, $\ell=P\lambda$, and entries \texttt{p,j:a} for the nonzero integers $a_{pj}$; the $\alpha_{pj}$ follow from the normalization of Section~\ref{sec:exact}, which is \cite[(39)]{FockLaughlin}. The rational quantities of \cite[Appendix~A]{FockLaughlin} are used unchanged, with the sums ranging over the enlarged index sets:
\begin{itemize}
\item $T\le T_3$ in the three-body formula \cite[(41)]{FockLaughlin};
\item $1\le D\le T_4$ in \cite[(43)--(46)]{FockLaughlin};
\item the allowance diagonal $(2\cdot\mathbf 1_{r=1}+\varepsilon_D)u_r$ in \cite[(46)]{FockLaughlin}.
\end{itemize}
The checks are those of Section~\ref{sec:exact}. Every three-body margin, for $z\in\{2,4,5,\dots,T_3\}$, must be positive, and every $M_D$, for $1\le D\le T_4$, positive semidefinite, both in exact arithmetic.

\subsection{A certificate for the original index sets}\label{app:ten}

The following ten rows over $P=10^{10}$ use the index sets of \cite{FockLaughlin}, so they can replace the seven rows of \cite[Appendix~A]{FockLaughlin} verbatim, with $\varepsilon_D=0$ for all $D$. The format is that of \cite[Appendix~A]{FockLaughlin}: the first two numbers are $t$ and $\ell$, and \texttt{pj:a} gives $a_{pj}$ with single-digit labels. They give
\[
  P^2\eta=92144309976470869849,\qquad \gamma=1-\eta-\frac{61}{4096}=\frac{6366432211029130151}{10^{20}}=0.0636643221\ldots,
\]
with all three-body margins positive (the smallest is $1.9\times10^{-8}$) and no negative characteristic coefficient in any block $D\le23$. The same rows are distributed as the ancillary file \nolinkurl{cert_7_8_8_7_10rows.json} in the directory \nolinkurl{other_certificates}.

\medskip\hrule\medskip
{\scriptsize\raggedright
\noindent
\texttt{0 6973372091 02:-3292091002 03:-1938404059 04:407409462 05:-4294194724 06:-1811079701 07:-2997699532}\\
\hspace*{1.5em}\texttt{08:-2047659164 11:1646045501 13:-299103437 14:-406127459 15:196205813 16:26246626 17:-826171532}\\
\hspace*{1.5em}\texttt{20:-823022751 21:484601015 22:149551719 23:-277731313 24:-85137944 25:171709340 26:-62485678}\\
\hspace*{1.5em}\texttt{30:-484601015 31:-50926183 32:379263178 33:-58252420 34:118829961 35:-60841783 40:1613415 41:129521514}\\
\hspace*{1.5em}\texttt{42:126706361 43:-167280272 44:64537941 50:-663689171 51:-69430499 52:68820519 53:-52944758}\\
\hspace*{1.5em}\texttt{60:-49899954 61:4704869 62:92413907 70:2000436 71:-47384710}\\
\texttt{1 4317556680 02:3488802039 03:1429965611 04:3162943949 05:4798902511 06:4880033395 07:12615114858}\\
\hspace*{1.5em}\texttt{08:25779293852 11:-1744401020 13:-939246695 14:-2104430234 15:-2488679546 16:-2686566331 17:3040978939}\\
\hspace*{1.5em}\texttt{18:-5713115660 20:872200510 21:-357491403 22:469623347 23:585873276 24:739629931 25:-618100576}\\
\hspace*{1.5em}\texttt{26:884283968 27:-1684886468 30:357491403 31:-395367994 32:-59765717 33:-121489295 34:41128090}\\
\hspace*{1.5em}\texttt{35:-227300049 36:-420980656 40:358240317 41:-6994769 42:-61403562 43:64636602 44:-47616615 45:55884459}\\
\hspace*{1.5em}\texttt{50:-57771634 51:157906044 52:285995345 53:-1102646 54:114816207 60:-242757704 61:-317104210}\\
\hspace*{1.5em}\texttt{62:34553883 63:32988509 70:-163447547 71:-523927307 72:130615560}\\
\texttt{2 2527712071 02:-1036977495 03:1705405901 04:2526145677 05:-1324844304 06:-15213516697 07:-10152453333}\\
\hspace*{1.5em}\texttt{08:-24343971898 11:518488748 13:2173711325 14:941708036 15:805635482 16:2018122472 17:-4395762570}\\
\hspace*{1.5em}\texttt{18:-391668375 20:-259244374 21:-426351475 22:-1086855663 23:-956141359 24:-523647978 25:-753767351}\\
\hspace*{1.5em}\texttt{26:-2555855154 27:-1787076362 28:3068341154 30:426351475 31:-315768210 32:720714350 33:365133898}\\
\hspace*{1.5em}\texttt{34:70509526 35:640127258 36:-363992991 37:3676064382 40:1017080146 41:-395267910 42:-255669299}\\
\hspace*{1.5em}\texttt{43:227155265 44:111062710 45:-38739249 46:3424181760 50:128678223 51:604559783 52:-309323797}\\
\hspace*{1.5em}\texttt{53:-59410222 54:124809470 55:2060613658 60:-1190179699 61:386756346 62:6782023 63:83415838}\\
\hspace*{1.5em}\texttt{64:971845300 70:-524497425 71:750242608 72:61005161 73:203415308}\\
\texttt{3 523234140 03:-800793685 04:2167465924 05:-587302276 06:-2310055218 07:11422000625 08:-4452113734}\\
\hspace*{1.5em}\texttt{13:-13938416 14:-2014230575 15:708511013 16:5692884949 17:-688450438 18:-8019454317 21:200198421}\\
\hspace*{1.5em}\texttt{22:6969208 23:-24843358 24:17354863 25:329944354 26:-849230418 27:2031045057 28:13303367998}\\
\hspace*{1.5em}\texttt{30:-200198421 31:-270933241 32:528401001 33:-131032831 34:-1106341817 35:534771276 36:1493865459}\\
\hspace*{1.5em}\texttt{37:12162309163 38:-1198574570 40:402915257 41:24284713 42:94969223 43:770387088 44:-163075354}\\
\hspace*{1.5em}\texttt{45:362897307 46:10814436722 47:-1243415214 50:-702859839 51:-80755024 52:-402038937 53:-3921031}\\
\hspace*{1.5em}\texttt{54:139533738 55:7485422622 56:-1013195845 60:100691178 61:-472187746 62:35925447 63:-815518104}\\
\hspace*{1.5em}\texttt{64:4621069228 65:-558909463 70:1884621092 71:105877718 72:-639956316 73:1746241662 74:-165555824}\\
\texttt{3 3271301073 03:-2481871512 04:6301117424 05:413667178 06:-428338615 07:11649737818 08:-30321199489}\\
\hspace*{1.5em}\texttt{13:-644792553 14:-1386438335 15:1095841852 16:-1741953627 17:-8673331163 18:4836122808 21:620467878}\\
\hspace*{1.5em}\texttt{22:322396276 23:1224084209 24:2777332977 25:439758883 26:-223812155 27:4531441027 28:-2369548605}\\
\hspace*{1.5em}\texttt{30:-620467878 31:-787639678 32:-877474625 33:-1665589722 34:-1195170022 35:-1241533884 36:100410338}\\
\hspace*{1.5em}\texttt{37:-2125855143 38:-626483132 40:1020261379 41:586187906 42:278404619 43:630272732 44:634841436}\\
\hspace*{1.5em}\texttt{45:438100278 46:-1135786248 47:-93220133 50:-381553582 51:-332768461 52:421202228 53:131185081}\\
\hspace*{1.5em}\texttt{54:32943005 55:171422240 56:-32731544 60:1315735143 61:-906889178 62:-22166040 63:-497485116}\\
\hspace*{1.5em}\texttt{64:581324040 65:190433867 70:630712476 71:659690248 72:-745675920 73:486553200 74:188756417}
\\
\noindent
\texttt{4 1413846164 04:-1793337943 05:72961489 06:13731952349 07:2524873851 08:23095321874 13:-255733550}\\
\hspace*{1.5em}\texttt{14:-629136384 15:1603172147 16:682440071 17:-2946521357 18:9965568789 22:127866775 23:-446285369}\\
\hspace*{1.5em}\texttt{24:-1397026875 25:-551877178 26:-2601895246 27:2572493947 28:-7682892539 31:224167243 32:603569465}\\
\hspace*{1.5em}\texttt{33:831230976 34:159963528 35:-1632206653 36:3916267912 37:-4859495360 38:1304755010 40:-400184252}\\
\hspace*{1.5em}\texttt{41:-227702777 42:-399472827 43:-151297860 44:-232035906 45:1194241305 46:-4632631625 47:155425024}\\
\hspace*{1.5em}\texttt{48:991271770 50:-187484934 51:-263226587 52:295936512 53:1265232087 54:-969644887 55:-2741040322}\\
\hspace*{1.5em}\texttt{56:739411774 57:409813041 60:1004181583 61:-370375195 62:362438920 63:-994466474 64:-1894899371}\\
\hspace*{1.5em}\texttt{65:675009972 66:581176756 70:322456351 71:-1476058210 72:-675591309 73:-722368847 74:499360213}\\
\hspace*{1.5em}\texttt{75:231809851}\\
\texttt{5 1621652124 05:-3071824663 06:6056290612 07:-12604090440 08:36323628211 14:-972115500 15:1034741848}\\
\hspace*{1.5em}\texttt{16:-2141500222 17:17029968454 18:-30190874680 23:20332639 24:-204583852 25:-194914056 26:1662975649}\\
\hspace*{1.5em}\texttt{27:-15020303114 28:-4660207283 32:222696236 33:-261699645 34:299800345 35:-1257595185 36:-9338875153}\\
\hspace*{1.5em}\texttt{37:-1327805116 38:735945805 41:202155361 42:416644626 43:144438521 44:100299093 45:-3142107078}\\
\hspace*{1.5em}\texttt{46:15903542 47:3495629909 48:-2019717835 50:-687764177 51:-644248545 52:-507276767 53:97717250}\\
\hspace*{1.5em}\texttt{54:1788863635 55:-772118887 56:2238342102 57:-619613011 58:-1196471940 60:927731127 61:1011887118}\\
\hspace*{1.5em}\texttt{62:25086719 63:3661645489 64:-497877461 65:1200374799 66:-297906784 67:-259709039 70:-1705525277}\\
\hspace*{1.5em}\texttt{71:-1526397059 72:2205327303 73:-31242682 74:98465866 75:-106307032 76:-163205199}\\
\texttt{6 669739959 06:-3340554299 07:7799399475 08:-12636725759 15:-986018007 16:2111037151 17:3290764420}\\
\hspace*{1.5em}\texttt{18:1483347901 24:257297 25:435601514 26:2793388238 27:-3257882178 28:11478103575 33:37477499}\\
\hspace*{1.5em}\texttt{34:-88892763 35:1700696930 36:-4624108021 37:14610635038 38:5648941371 42:66843030 43:-147273429}\\
\hspace*{1.5em}\texttt{44:-206957157 45:-2015319067 46:11881823216 47:5368868574 48:1560961074 51:151400006 52:78145423}\\
\hspace*{1.5em}\texttt{53:-788565889 54:910206734 55:7759891516 56:3918463509 57:786598336 58:601581565 60:-561721806}\\
\hspace*{1.5em}\texttt{61:-385990506 62:-612833759 63:1692298488 64:3114517747 65:1860403234 66:378834976 67:336084574}\\
\hspace*{1.5em}\texttt{68:1391725165 70:1201786479 71:1375688133 72:1470903817 73:1039668908 74:984060859 75:489866378}\\
\hspace*{1.5em}\texttt{76:431805775 77:797393052}\\
\texttt{6 214266307 06:-439467806 07:-1375529908 08:-3435921511 15:-49667205 16:-417315664 17:456170941}\\
\hspace*{1.5em}\texttt{18:-818175167 24:29799602 25:187705245 26:491257031 27:927957446 28:-1291983329 33:2244531}\\
\hspace*{1.5em}\texttt{34:144827460 35:570565341 36:-360917212 37:-2230782745 38:4660139336 42:-19508097 43:-177148038}\\
\hspace*{1.5em}\texttt{44:-151756301 45:-394581870 46:-1402776017 47:4091122160 48:-1845424335 51:35163784 52:69284493}\\
\hspace*{1.5em}\texttt{53:-135312191 54:328371393 55:-1007392144 56:2700709233 57:-1032259631 58:562090675 60:-49705156}\\
\hspace*{1.5em}\texttt{61:18148689 62:-163308611 63:1896421 64:-870384939 65:1568117439 66:-979037501 67:371914911}\\
\hspace*{1.5em}\texttt{68:35376380 70:-35842740 71:352353286 72:-73475459 73:-221618153 74:395382270 75:-353404236}\\
\hspace*{1.5em}\texttt{76:141396328 77:-78752196}\\
\texttt{7 1545122886 07:-6641430073 08:18391240234 16:-1748950509 17:10690255140 18:22922119052 25:-103640818}\\
\hspace*{1.5em}\texttt{26:3251523769 27:20428378460 28:4230493110 34:147341859 35:511898835 36:13156779869 37:4038383518}\\
\hspace*{1.5em}\texttt{38:-1869182592 43:-27407472 44:-532911138 45:4238118397 46:2420464818 47:-842456500 48:2354108225}\\
\hspace*{1.5em}\texttt{52:-5068028 53:-297383744 54:-2336295740 55:1646529445 56:-71628095 57:821673728 58:332332926}\\
\hspace*{1.5em}\texttt{61:278491290 62:85692010 63:-4761270096 64:604165262 65:-788200009 66:545674174 67:-48713485}\\
\hspace*{1.5em}\texttt{68:-905612214 70:-839323072 71:-973690547 72:-4073336878 73:-910054710 74:-930476394 75:-27474128}\\
\hspace*{1.5em}\texttt{76:-543971027 77:-649302971 78:-749507015}
\par}
\medskip\hrule\medskip

\subsection{The main certificate}\label{app:main}

The certificate behind Theorem~\ref{thm:main} has $(\pbar,\jbar,\ibar,\tbar)=(17,18,18,17)$ and $64$ rows over $P=10^{13}$, with $14417$ nonzero integers $a_{pj}$. It is too long to print, so it is distributed as the ancillary file \texttt{certificate.json}, with SHA-256 hash
\[
  \texttt{b9a94d46c8a90e0206d98b700712686d3f2bc2725efcadab3e6513759abbba81},
\]
together with the same rows in the text format above (\texttt{rows.txt}, labels written \texttt{p,j:a}). A self-contained verification script, \texttt{verify.py}, uses only the Python standard library and reads only the certificate file. It recomputes every condition of Definition~\ref{def:admissible} in exact integer and rational arithmetic, exits with status zero if and only if all of them hold and the recomputed values agree with those recorded in the file, and prints the exact constant. It runs in under a minute. Run on the seven rows of \cite{FockLaughlin}, it reproduces their constant, all thirteen floors of \cite[Appendix~A.2]{FockLaughlin} and all $23$ sign strings of \cite[Appendix~A.3]{FockLaughlin}. The data of the certificate are:
\begin{itemize}
\item rows per $t=0,\dots,17$: $3,3,4,4,2,2,2,3,4,3,3,3,3,1,5,2,2,15$;
\item allowances $\varepsilon_D=0$ except $\varepsilon_{13}=20271/(4\cdot10^{11})$, $\varepsilon_{14}=39507/(2\cdot10^{12})$, $\varepsilon_{15}=866443/10^{13}$, $\varepsilon_{17}=9649/(125\cdot10^9)$, $\varepsilon_{20}=72481/(5\cdot10^{12})$, $\varepsilon_{37}=129377/10^{13}$ and $\varepsilon_{39}=56781/10^{13}$;
\item $T_3=35$, $T_4=53$, and local modes $0,\dots,54$ in the transfer lemma;
\item all $33$ three-body margins positive, the smallest $1.99999\times10^{-8}$ at $z=5$;
\item no negative characteristic coefficient in any block $D\le53$;
\item $\eta$ and $\gamma_*$ as in~\eqref{eq:gammastar}.
\end{itemize}
Four exact implementations agree on every quantity:
\begin{itemize}
\item \texttt{verify.py};
\item \texttt{sdp/exact\_verify.py}, the exact checker of the search code, written independently;
\item a literal implementation of \cite[(40)--(47)]{FockLaughlin};
\item a separately written checker built from the defining formulas.
\end{itemize}
The first two are among the ancillary files (Appendix~\ref{app:data}); the other two are not distributed.
Perturbing all its integers by a relative $10^{-6}$, or dropping any row, makes the check fail, and finite-flux coupled bases converge to the limiting data for the largest blocks used (Section~\ref{sec:exact}). The other certificates of Tables~\ref{tab:ledger} and~\ref{tab:trend}, which also pass \texttt{verify.py}, and the code of the SDP search are listed in Appendix~\ref{app:data}.

\section{Data and code}\label{app:data}

The following ancillary files are distributed with this paper on arXiv (directory \texttt{anc}):
\begin{itemize}
\item \texttt{certificate.json}: the main certificate of Theorem~\ref{thm:main}, described in Appendix~\ref{app:main};
\item \texttt{rows.txt}: the same rows in the text format of Appendix~\ref{app:cert};
\item \texttt{verify.py}: the self-contained exact verifier of Appendix~\ref{app:main};
\item \texttt{verify\_output.txt}: the recorded output of \texttt{verify.py} on the main certificate;
\item \texttt{verify\_all.py} and its output \texttt{verify\_all\_output.txt}: a script that runs \texttt{verify.py} on all fourteen certificate files, namely \texttt{certificate.json}, \texttt{paper\_rows.json} and the twelve \texttt{.json} files in \texttt{other\_certificates/};
\item \texttt{variants.py}: a script that computes the constants with the product-twirl tail of Lemma~\ref{lem:tail}, the staircase price of Proposition~\ref{prop:four}(b), or both (Tables~\ref{tab:ledger} and~\ref{tab:tail}, Corollary~\ref{cor:144} and Section~\ref{sec:proofs});
\item \texttt{paper\_rows.json}: the seven rows of \cite{FockLaughlin}, on which \texttt{verify.py} reproduces the original constant $4616733319001/10^{14}$;
\item \texttt{other\_certificates/}: the ten-row certificate for the original index sets (Appendix~\ref{app:ten}), also as a text file, and the certificates for $(8,9,9,8)$, $(9,10,10,9)$, $(9,10,13,9)$, $(10,11,11,10)$, $(11,12,12,11)$, $(12,13,13,12)$, $(13,14,14,13)$, $(14,15,15,14)$, $(15,16,16,15)$, $(15,16,18,15)$ and $(16,17,17,16)$ of Tables~\ref{tab:ledger} and~\ref{tab:trend}, all in the format of \texttt{certificate.json} and all passing \texttt{verify.py}; for $(16,17,17,16)$ also the rows as a text file and the output of \texttt{verify.py};
\item \texttt{sdp/}: the code of the SDP search of Section~\ref{sec:sdp} (data assembly, the interior-point solver and the conversion to exact certificates of Section~\ref{sec:exact}), which needs NumPy and SciPy and has its own \texttt{README.md}; it is not needed to verify the certificates;
\item \texttt{SHA256SUMS.txt}: the SHA-256 hashes of all files;
\item \texttt{README.md}: a description of the files.
\end{itemize}
The author declares no copyright on these AI-generated data and code. The seven rows in \texttt{paper\_rows.json} are transcribed from \cite[Appendix~A]{FockLaughlin} and are not covered by this statement.

\section{Human input}\label{app:human}

After OpenAI released its proof \cite{FockLaughlin} and the Lean formalization \cite{FockLaughlinLean}, I went through the manuscript and realized that the certificate it used appeared to be hand-picked rather than obtained by optimization. This left considerable room for improvement if the certificate could be chosen optimally. I therefore asked Claude Opus 5.5 to set up a semidefinite-programming search for the optimal certificate, by formulating the correct constraints. After the present manuscript was finished, I also asked ChatGPT 6.1 Sol to produce a Lean-formalized proof of the improved bound; this is in progress and will be included in a later version.


\begin{thebibliography}{99}\small\raggedright

\bibitem{Laughlin1983}
Robert B. Laughlin. Anomalous quantum Hall effect: an incompressible quantum fluid with fractionally charged excitations. \emph{Physical Review Letters}, 50(18):1395--1398, 1983. \href{https://doi.org/10.1103/PhysRevLett.50.1395}{doi:10.1103/\allowbreak PhysRevLett.50.1395}.

\bibitem{Haldane1983}
F.~D.~M. Haldane. Fractional quantization of the Hall effect: a hierarchy of incompressible quantum fluid states. \emph{Physical Review Letters}, 51(7):605--608, 1983. \href{https://doi.org/10.1103/PhysRevLett.51.605}{doi:10.1103/PhysRevLett.51.605}.

\bibitem{HaldaneRezayi1985}
F.~D.~M. Haldane and E.~H. Rezayi. Finite-size studies of the incompressible state of the fractionally quantized Hall effect and its excitations. \emph{Physical Review Letters}, 54(3):237--240, 1985. \href{https://doi.org/10.1103/PhysRevLett.54.237}{doi:10.1103/PhysRevLett.54.237}.

\bibitem{Rougerie2019}
Nicolas Rougerie. On the Laughlin function and its perturbations. \emph{S\'eminaire Laurent Schwartz --- EDP et applications (2018--2019)}, Expos\'e no.~II, 17~pp., 2019. \href{https://doi.org/10.5802/slsedp.131}{doi:10.5802/slsedp.131}.

\bibitem{FockLaughlin}
OpenAI. A Fock-space inequality and the Laughlin spectral gap. OpenAI Math Release preprint, September 24, 2026. \url{https://github.com/openai/math/blob/main/preprints/A-Fock-space-inequality-and-the-Laughlin-spectral-gap-September-24-2026/A-Fock-space-inequality-and-the-Laughlin-spectral-gap-September-24-2026.pdf}.

\bibitem{FockLaughlinLean}
OpenAI. Lean formalization of ``A Fock-space inequality and the Laughlin spectral gap''. GitHub repository \texttt{openai/math}, 2026. \url{https://github.com/openai/math/tree/main/lean/OAI/Analysis/LaughlinFock}.

\bibitem{VandenbergheBoyd1996}
Lieven Vandenberghe and Stephen Boyd. Semidefinite programming. \emph{SIAM Review}, 38(1):49--95, 1996. \href{https://doi.org/10.1137/1038003}{doi:10.1137/1038003}.

\bibitem{Gao2025}
Qiang Gao, Ryan A. Lanzetta, Patrick Ledwith, Jie Wang, and Eslam Khalaf. Bootstrapping the quantum Hall problem. \emph{Physical Review X}, 15(3):031034, 2025. \href{https://doi.org/10.1103/csnn-vjhn}{doi:10.1103/csnn-vjhn}.

\bibitem{HKM1996}
Christoph Helmberg, Franz Rendl, Robert J. Vanderbei, and Henry Wolkowicz. An interior-point method for semidefinite programming. \emph{SIAM Journal on Optimization}, 6(2):342--361, 1996. \href{https://doi.org/10.1137/0806020}{doi:10.1137/0806020}.

\bibitem{KSH1997}
Masakazu Kojima, Susumu Shindoh, and Shinji Hara. Interior-point methods for the monotone semidefinite linear complementarity problem in symmetric matrices. \emph{SIAM Journal on Optimization}, 7(1):86--125, 1997. \href{https://doi.org/10.1137/S1052623494269035}{doi:10.1137/S1052623494269035}.

\bibitem{Monteiro1997}
Renato D.~C. Monteiro. Primal--dual path-following algorithms for semidefinite programming. \emph{SIAM Journal on Optimization}, 7(3):663--678, 1997. \href{https://doi.org/10.1137/S1052623495293056}{doi:10.1137/S1052623495293056}.

\bibitem{Mehrotra1992}
Sanjay Mehrotra. On the implementation of a primal-dual interior point method. \emph{SIAM Journal on Optimization}, 2(4):575--601, 1992. \href{https://doi.org/10.1137/0802028}{doi:10.1137/0802028}.

\bibitem{Hou1998}
Shui-Hung Hou. A simple proof of the Leverrier--Faddeev characteristic polynomial algorithm. \emph{SIAM Review}, 40(3):706--709, 1998. \href{https://doi.org/10.1137/S003614459732076X}{doi:10.1137/S003614459732076X}.

\bibitem{PeyrlParrilo2008}
Helfried Peyrl and Pablo A. Parrilo. Computing sum of squares decompositions with rational coefficients. \emph{Theoretical Computer Science}, 409(2):269--281, 2008. \href{https://doi.org/10.1016/j.tcs.2008.09.025}{doi:10.1016/j.tcs.2008.09.025}.

\bibitem{LNWY2026}
Marius Lemm, Bruno Nachtergaele, Simone Warzel, and Amanda Young. Recursive spectral relations and the charge vs neutral gap in fractional quantum Hall systems. \emph{Journal of Mathematical Physics}, 67(1):011901, 2026. \href{https://doi.org/10.1063/5.0292975}{doi:10.1063/5.0292975}.

\end{thebibliography}
\end{document}